%% file: 18ICML_PhaseInit.tex
\documentclass{article}

\usepackage{times}
\usepackage{graphicx} 
\usepackage{natbib}
\usepackage{algorithm}
\usepackage{algorithmic}
\usepackage{hyperref}

\usepackage[accepted]{icml2018}

\usepackage{url}            
\usepackage{nicefrac}       
\usepackage{amssymb}
\usepackage{amsmath}
\usepackage{amsthm}
\usepackage{paralist}
\usepackage{subcaption}
\usepackage{setspace}
\usepackage{booktabs}
\usepackage{microtype}

\usepackage{balance}

\input{vmr-symbols-vecbold}
\input{standard-macros}
\input{defs}

\allowdisplaybreaks

\usepackage{color}

\icmltitlerunning{Linear Spectral Estimators and an Application to Phase Retrieval}

\begin{document}

\twocolumn[
\icmltitle{Linear Spectral Estimators and an Application to Phase Retrieval}

\icmlsetsymbol{equal}{*}

\begin{icmlauthorlist}
	\icmlauthor{Ramina Ghods}{cu}
	\icmlauthor{Andrew S. Lan}{pr}
	\icmlauthor{Tom Goldstein}{mr}
	\icmlauthor{Christoph Studer}{cu}
\end{icmlauthorlist}

\icmlaffiliation{cu}{School~of Electrical and Computer Engineering, Cornell University, Ithaca, NY}
\icmlaffiliation{pr}{Department of EE, Princeton University}
\icmlaffiliation{mr}{University of Maryland}

\icmlcorrespondingauthor{Ramina Ghods}{rg548@cornell.edu}
\icmlcorrespondingauthor{Christoph Studer}{studer@cornell.edu}

\icmlkeywords{[phase retrieval], machine learning, ICML}

\vskip 0.3in
]



\printAffiliationsAndNotice{} 

%

\begin{abstract}
Phase retrieval refers to the problem of recovering real-  or complex-valued vectors from magnitude measurements.
The best-known algorithms for this problem are iterative in nature and rely on so-called spectral initializers that provide accurate initialization vectors.
We propose a novel class of estimators suitable for general nonlinear measurement systems, called linear spectral estimators~(LSPEs), which can be used to compute accurate initialization vectors for phase retrieval problems.
The proposed LSPEs not only provide accurate initialization vectors for noisy phase retrieval systems with structured or random measurement matrices, but also enable the derivation of sharp and nonasymptotic mean-squared error bounds.
We demonstrate the efficacy of LSPEs on synthetic and real-world phase retrieval problems, and show that our estimators significantly outperform existing methods for structured measurement systems that arise in practice. 
\end{abstract}




\input{1-introduction.tex}
\input{2-estimator.tex}
\input{3-phaseretrieval.tex}
\input{4-results.tex}

\input{5-conclusions.tex}


\section*{Acknowledgments}
R. Ghods and C.~Studer were supported in part by Xilinx, Inc.\ and by the US National Science Foundation (NSF) under grants ECCS-1408006, EECS-1740286, CCF-1535897,  CCF-1652065, and CNS-1717559.
T. Goldstein was supported by the US NSF under grant CCF-1535902, the US ONR under grant N00014-15-1-2676, the DARPA Lifelong Learning Machines program, and the Sloan Foundation. 


\appendix
\input{6-appendices-smushed.tex}
\input{7-supplementary.tex}

\bibliographystyle{icml2018}

\bibliography{confs-jrnls,publishers,phase}

\balance

\end{document}

%% file: vmr-symbols-vecbold.tex
\usepackage{amssymb}
\usepackage{amsfonts}
\usepackage{mathrsfs}
\usepackage{xspace}
\usepackage{bm}
\usepackage{upgreek}

\newcommand{\safemath}[2]{\newcommand{#1}{\ensuremath{#2}\xspace}}

\safemath{\bma}{\mathbf{a}}
\safemath{\bmb}{\mathbf{b}}
\safemath{\bmc}{\mathbf{c}}
\safemath{\bmd}{\mathbf{d}}
\safemath{\bme}{\mathbf{e}}
\safemath{\bmf}{\mathbf{f}}
\safemath{\bmg}{\mathbf{g}}
\safemath{\bmh}{\mathbf{h}}
\safemath{\bmi}{\mathbf{i}}
\safemath{\bmj}{\mathbf{j}}
\safemath{\bmk}{\mathbf{k}}
\safemath{\bml}{\mathbf{l}}
\safemath{\bmm}{\mathbf{m}}
\safemath{\bmn}{\mathbf{n}}
\safemath{\bmo}{\mathbf{o}}
\safemath{\bmp}{\mathbf{p}}
\safemath{\bmq}{\mathbf{q}}
\safemath{\bmr}{\mathbf{r}}
\safemath{\bms}{\mathbf{s}}
\safemath{\bmt}{\mathbf{t}}
\safemath{\bmu}{\mathbf{u}}
\safemath{\bmv}{\mathbf{v}}
\safemath{\bmw}{\mathbf{w}}
\safemath{\bmx}{\mathbf{x}}
\safemath{\bmy}{\mathbf{y}}
\safemath{\bmz}{\mathbf{z}}
\safemath{\bmzero}{\mathbf{0}}
\safemath{\bmone}{\mathbf{1}}

\bmdefine{\biad}{a}
\bmdefine{\bibd}{b}
\bmdefine{\bicd}{c}
\bmdefine{\bidd}{d}
\bmdefine{\bied}{e}
\bmdefine{\bifd}{f}
\bmdefine{\bigd}{g}
\bmdefine{\bihd}{h}
\bmdefine{\biid}{i}
\bmdefine{\bijd}{j}
\bmdefine{\bikd}{k}
\bmdefine{\bild}{l}
\bmdefine{\bimd}{m}
\bmdefine{\bind}{n}
\bmdefine{\biod}{o}
\bmdefine{\bipd}{p}
\bmdefine{\biqd}{q}
\bmdefine{\bird}{r}
\bmdefine{\bisd}{s}
\bmdefine{\bitd}{t}
\bmdefine{\biud}{u}
\bmdefine{\bivd}{v}
\bmdefine{\biwd}{w}
\bmdefine{\bixd}{x}
\bmdefine{\biyd}{y}
\bmdefine{\bizd}{z}

\bmdefine{\bixid}{\xi}
\bmdefine{\bilambdad}{\lambda}
\bmdefine{\bimud}{\mu}
\bmdefine{\bithetad}{\theta}
\bmdefine{\biphid}{\phi}
\bmdefine{\bideltad}{\delta}

\safemath{\bmia}{\biad}
\safemath{\bmib}{\bibd}
\safemath{\bmic}{\bicd}
\safemath{\bmid}{\bidd}
\safemath{\bmie}{\bied}
\safemath{\bmif}{\bifd}
\safemath{\bmig}{\bigd}
\safemath{\bmih}{\bihd}
\safemath{\bmii}{\biid}
\safemath{\bmij}{\bijd}
\safemath{\bmik}{\bikd}
\safemath{\bmil}{\bild}
\safemath{\bmim}{\bimd}
\safemath{\bmin}{\bind}
\safemath{\bmio}{\biod}
\safemath{\bmip}{\bipd}
\safemath{\bmiq}{\biqd}
\safemath{\bmir}{\bird}
\safemath{\bmis}{\bisd}
\safemath{\bmit}{\bitd}
\safemath{\bmiu}{\biud}
\safemath{\bmiv}{\bivd}
\safemath{\bmiw}{\biwd}
\safemath{\bmix}{\bixd}
\safemath{\bmiy}{\biyd}
\safemath{\bmiz}{\bizd}

\safemath{\bmxi}{\bixid}
\safemath{\bmlambda}{\bilambdad}
\safemath{\bmmu}{\bimud}
\safemath{\bmtheta}{\bithetad}
\safemath{\bmphi}{\biphid}
\safemath{\bmdelta}{\bideltad}

\safemath{\bA}{\mathbf{A}}
\safemath{\bB}{\mathbf{B}}
\safemath{\bC}{\mathbf{C}}
\safemath{\bD}{\mathbf{D}}
\safemath{\bE}{\mathbf{E}}
\safemath{\bF}{\mathbf{F}}
\safemath{\bG}{\mathbf{G}}
\safemath{\bH}{\mathbf{H}}
\safemath{\bI}{\mathbf{I}}
\safemath{\bJ}{\mathbf{J}}
\safemath{\bK}{\mathbf{K}}
\safemath{\bL}{\mathbf{L}}
\safemath{\bM}{\mathbf{M}}
\safemath{\bN}{\mathbf{N}}
\safemath{\bO}{\mathbf{O}}
\safemath{\bP}{\mathbf{P}}
\safemath{\bQ}{\mathbf{Q}}
\safemath{\bR}{\mathbf{R}}
\safemath{\bS}{\mathbf{S}}
\safemath{\bT}{\mathbf{T}}
\safemath{\bU}{\mathbf{U}}
\safemath{\bV}{\mathbf{V}}
\safemath{\bW}{\mathbf{W}}
\safemath{\bX}{\mathbf{X}}
\safemath{\bY}{\mathbf{Y}}
\safemath{\bZ}{\mathbf{Z}}

\safemath{\bZero}{\mathbf{0}}
\safemath{\bOne}{\mathbf{1}}
\safemath{\bDelta}{\mathbf{\Delta}}
\safemath{\bLambda}{\boldsymbol\Lambda}
\safemath{\bPhi}{\mathbf{\Upphi}}
\safemath{\bSigma}{\mathbf{\Upsigma}}
\safemath{\bOmega}{\mathbf{\Upomega}}
\safemath{\bTheta}{\mathbf{\Uptheta}}

\bmdefine{\biAd}{A}
\bmdefine{\biBd}{B}
\bmdefine{\biCd}{C}
\bmdefine{\biDd}{D}
\bmdefine{\biEd}{E}
\bmdefine{\biFd}{F}
\bmdefine{\biGd}{G}
\bmdefine{\biHd}{H}
\bmdefine{\biId}{I}
\bmdefine{\biJd}{J}
\bmdefine{\biKd}{K}
\bmdefine{\biLd}{L}
\bmdefine{\biMd}{M}
\bmdefine{\biOd}{N}
\bmdefine{\biPd}{O}
\bmdefine{\biQd}{P}
\bmdefine{\biRd}{R}
\bmdefine{\biSd}{S}
\bmdefine{\biTd}{T}
\bmdefine{\biUd}{U}
\bmdefine{\biVd}{V}
\bmdefine{\biWd}{W}
\bmdefine{\biXd}{X}
\bmdefine{\biYd}{Y}
\bmdefine{\biZd}{Z}

\bmdefine{\biDelta}{\Delta}
\bmdefine{\biLambda}{\Lambda}
\bmdefine{\biPhi}{\Phi}
\bmdefine{\biSigma}{\Sigma}
\bmdefine{\biOmega}{\Omega}
\bmdefine{\biTheta}{\Theta}

\safemath{\bimA}{\biAd}
\safemath{\bimB}{\biBd}
\safemath{\bimC}{\biCd}
\safemath{\bimD}{\biDd}
\safemath{\bimE}{\biEd}
\safemath{\bimF}{\biFd}
\safemath{\bimG}{\biGd}
\safemath{\bimH}{\biHd}
\safemath{\bimI}{\biId}
\safemath{\bimJ}{\biJd}
\safemath{\bimK}{\biKd}
\safemath{\bimL}{\biLd}
\safemath{\bimM}{\biMd}
\safemath{\bimN}{\biNd}
\safemath{\bimO}{\biOd}
\safemath{\bimP}{\biPd}
\safemath{\bimQ}{\biQd}
\safemath{\bimR}{\biRd}
\safemath{\bimS}{\biSd}
\safemath{\bimT}{\biTd}
\safemath{\bimU}{\biUd}
\safemath{\bimV}{\biVd}
\safemath{\bimW}{\biWd}
\safemath{\bimX}{\biXd}
\safemath{\bimY}{\biYd}
\safemath{\bimZ}{\biZd}

\safemath{\bimDelta}{\biDelta}
\safemath{\bimLambda}{\biLambda}
\safemath{\bimPhi}{\biPhi}
\safemath{\bimSigma}{\biSigma}
\safemath{\bimOmega}{\biOmega}
\safemath{\bimTheta}{\biTheta}

\safemath{\setA}{\mathcal{A}}
\safemath{\setB}{\mathcal{B}}
\safemath{\setC}{\mathcal{C}}
\safemath{\setD}{\mathcal{D}}
\safemath{\setE}{\mathcal{E}}
\safemath{\setF}{\mathcal{F}}
\safemath{\setG}{\mathcal{G}}
\safemath{\setH}{\mathcal{H}}
\safemath{\setI}{\mathcal{I}}
\safemath{\setJ}{\mathcal{J}}
\safemath{\setK}{\mathcal{K}}
\safemath{\setL}{\mathcal{L}}
\safemath{\setM}{\mathcal{M}}
\safemath{\setN}{\mathcal{N}}
\safemath{\setO}{\mathcal{O}}
\safemath{\setP}{\mathcal{P}}
\safemath{\setQ}{\mathcal{Q}}
\safemath{\setR}{\mathcal{R}}
\safemath{\setS}{\mathcal{S}}
\safemath{\setT}{\mathcal{T}}
\safemath{\setU}{\mathcal{U}}
\safemath{\setV}{\mathcal{V}}
\safemath{\setW}{\mathcal{W}}
\safemath{\setX}{\mathcal{X}}
\safemath{\setY}{\mathcal{Y}}
\safemath{\setZ}{\mathcal{Z}}
\safemath{\emptySet}{\varnothing}

\safemath{\colA}{\mathscr{A}}
\safemath{\colB}{\mathscr{B}}
\safemath{\colC}{\mathscr{C}}
\safemath{\colD}{\mathscr{D}}
\safemath{\colE}{\mathscr{E}}
\safemath{\colF}{\mathscr{F}}
\safemath{\colG}{\mathscr{G}}
\safemath{\colH}{\mathscr{H}}
\safemath{\colI}{\mathscr{I}}
\safemath{\colJ}{\mathscr{J}}
\safemath{\colK}{\mathscr{K}}
\safemath{\colL}{\mathscr{L}}
\safemath{\colM}{\mathscr{M}}
\safemath{\colN}{\mathscr{N}}
\safemath{\colO}{\mathscr{O}}
\safemath{\colP}{\mathscr{P}}
\safemath{\colQ}{\mathscr{Q}}
\safemath{\colR}{\mathscr{R}}
\safemath{\colS}{\mathscr{S}}
\safemath{\colT}{\mathscr{T}}
\safemath{\colU}{\mathscr{U}}
\safemath{\colV}{\mathscr{V}}
\safemath{\colW}{\mathscr{W}}
\safemath{\colX}{\mathscr{X}}
\safemath{\colY}{\mathscr{Y}}
\safemath{\colZ}{\mathscr{Z}}

\safemath{\opA}{\mathbb{A}}
\safemath{\opB}{\mathbb{B}}
\safemath{\opC}{\mathbb{C}}
\safemath{\opD}{\mathbb{D}}
\safemath{\opE}{\mathbb{E}}
\safemath{\opF}{\mathbb{F}}
\safemath{\opG}{\mathbb{G}}
\safemath{\opH}{\mathbb{H}}
\safemath{\opI}{\mathbb{I}}
\safemath{\opJ}{\mathbb{J}}
\safemath{\opK}{\mathbb{K}}
\safemath{\opL}{\mathbb{L}}
\safemath{\opM}{\mathbb{M}}
\safemath{\opN}{\mathbb{N}}
\safemath{\opO}{\mathbb{O}}
\safemath{\opP}{\mathbb{P}}
\safemath{\opQ}{\mathbb{Q}}
\safemath{\opR}{\mathbb{R}}
\safemath{\opS}{\mathbb{S}}
\safemath{\opT}{\mathbb{T}}
\safemath{\opU}{\mathbb{U}}
\safemath{\opV}{\mathbb{V}}
\safemath{\opW}{\mathbb{W}}
\safemath{\opX}{\mathbb{X}}
\safemath{\opY}{\mathbb{Y}}
\safemath{\opZ}{\mathbb{Z}}
\safemath{\opZero}{\mathbb{O}}
\safemath{\identityop}{\opI}

\safemath{\veca}{\bma}
\safemath{\vecb}{\bmb}
\safemath{\vecc}{\bmc}
\safemath{\vecd}{\bmd}
\safemath{\vece}{\bme}
\safemath{\vecf}{\bmf}
\safemath{\vecg}{\bmg}
\safemath{\vech}{\bmh}
\safemath{\veci}{\bmi}
\safemath{\vecj}{\bmj}
\safemath{\veck}{\bmk}
\safemath{\vecl}{\bml}
\safemath{\vecm}{\bmm}
\safemath{\vecn}{\bmn}
\safemath{\veco}{\bmo}
\safemath{\vecp}{\bmp}
\safemath{\vecq}{\bmq}
\safemath{\vecr}{\bmr}
\safemath{\vecs}{\bms}
\safemath{\vect}{\bmt}
\safemath{\vecu}{\bmu}
\safemath{\vecv}{\bmv}
\safemath{\vecw}{\bmw}
\safemath{\vecx}{\bmx}
\safemath{\vecy}{\bmy}
\safemath{\vecz}{\bmz}

\safemath{\veczero}{\bmzero}
\safemath{\vecone}{\bmone}
\safemath{\vecxi}{\bmxi}
\safemath{\veclambda}{\bmlambda}
\safemath{\vecmu}{\bmmu}
\safemath{\vectheta}{\bmtheta}
\safemath{\vecphi}{\bmphi}
\safemath{\vecdelta}{\bmdelta}

\safemath{\matA}{\bA}
\safemath{\matB}{\bB}
\safemath{\matC}{\bC}
\safemath{\matD}{\bD}
\safemath{\matE}{\bE}
\safemath{\matF}{\bF}
\safemath{\matG}{\bG}
\safemath{\matH}{\bH}
\safemath{\matI}{\bI}
\safemath{\matJ}{\bJ}
\safemath{\matK}{\bK}
\safemath{\matL}{\bL}
\safemath{\matM}{\bM}
\safemath{\matN}{\bN}
\safemath{\matO}{\bO}
\safemath{\matP}{\bP}
\safemath{\matQ}{\bQ}
\safemath{\matR}{\bR}
\safemath{\matS}{\bS}
\safemath{\matT}{\bT}
\safemath{\matU}{\bU}
\safemath{\matV}{\bV}
\safemath{\matW}{\bW}
\safemath{\matX}{\bX}
\safemath{\matY}{\bY}
\safemath{\matZ}{\bZ}
\safemath{\matzero}{\bmzero}

\safemath{\matDelta}{\bDelta}
\safemath{\matLambda}{\bLambda}
\safemath{\matPhi}{\bPhi}
\safemath{\matSigma}{\bSigma}
\safemath{\matOmega}{\bOmega}
\safemath{\matTheta}{\bTheta}

\safemath{\matidentity}{\matI}
\safemath{\matone}{\matO}

\safemath{\rnda}{A}
\safemath{\rndb}{B}
\safemath{\rndc}{C}
\safemath{\rndd}{D}
\safemath{\rnde}{E}
\safemath{\rndf}{F}
\safemath{\rndg}{G}
\safemath{\rndh}{H}
\safemath{\rndi}{I}
\safemath{\rndj}{J}
\safemath{\rndk}{K}
\safemath{\rndl}{L}
\safemath{\rndm}{M}
\safemath{\rndn}{N}
\safemath{\rndo}{O}
\safemath{\rndp}{P}
\safemath{\rndq}{Q}
\safemath{\rndr}{R}
\safemath{\rnds}{S}
\safemath{\rndt}{T}
\safemath{\rndu}{U}
\safemath{\rndv}{V}
\safemath{\rndw}{W}
\safemath{\rndx}{X}
\safemath{\rndy}{Y}
\safemath{\rndz}{Z}

\safemath{\rveca}{\bimA}
\safemath{\rvecb}{\bimB}
\safemath{\rvecc}{\bimC}
\safemath{\rvecd}{\bimD}
\safemath{\rvece}{\bimE}
\safemath{\rvecf}{\bimF}
\safemath{\rvecg}{\bimG}
\safemath{\rvech}{\bimH}
\safemath{\rveci}{\bimI}
\safemath{\rvecj}{\bimJ}
\safemath{\rveck}{\bimK}
\safemath{\rvecl}{\bimL}
\safemath{\rvecm}{\bimM}
\safemath{\rvecn}{\bimN}
\safemath{\rveco}{\bomO}
\safemath{\rvecp}{\bimP}
\safemath{\rvecq}{\bimQ}
\safemath{\rvecr}{\bimR}
\safemath{\rvecs}{\bimS}
\safemath{\rvect}{\bimT}
\safemath{\rvecu}{\bimU}
\safemath{\rvecv}{\bimV}
\safemath{\rvecw}{\bimW}
\safemath{\rvecx}{\bimX}
\safemath{\rvecy}{\bimY}
\safemath{\rvecz}{\bimZ}

\safemath{\rvecxi}{\bmxi}
\safemath{\rveclambda}{\bmlambda}
\safemath{\rvecmu}{\bmmu}
\safemath{\rvectheta}{\bmtheta}
\safemath{\rvecphi}{\bmphi}

\safemath{\rmatA}{\bimA}
\safemath{\rmatB}{\bimB}
\safemath{\rmatC}{\bimC}
\safemath{\rmatD}{\bimD}
\safemath{\rmatE}{\bimE}
\safemath{\rmatF}{\bimF}
\safemath{\rmatG}{\bimG}
\safemath{\rmatH}{\bimH}
\safemath{\rmatI}{\bimI}
\safemath{\rmatJ}{\bimJ}
\safemath{\rmatK}{\bimK}
\safemath{\rmatL}{\bimL}
\safemath{\rmatM}{\bimM}
\safemath{\rmatN}{\bimN}
\safemath{\rmatO}{\bimO}
\safemath{\rmatP}{\bimP}
\safemath{\rmatQ}{\bimQ}
\safemath{\rmatR}{\bimR}
\safemath{\rmatS}{\bimS}
\safemath{\rmatT}{\bimT}
\safemath{\rmatU}{\bimU}
\safemath{\rmatV}{\bimV}
\safemath{\rmatW}{\bimW}
\safemath{\rmatX}{\bimX}
\safemath{\rmatY}{\bimY}
\safemath{\rmatZ}{\bimZ}

\safemath{\rmatDelta}{\bimDelta}
\safemath{\rmatLambda}{\bimLambda}
\safemath{\rmatPhi}{\bimPhi}
\safemath{\rmatSigma}{\bimSigma}
\safemath{\rmatOmega}{\bimOmega}
\safemath{\rmatTheta}{\bimTheta}

%% file: standard-macros.tex
\usepackage{amssymb}
\usepackage{amsfonts}
\usepackage{mathrsfs}
\usepackage{xspace}
\usepackage{bm}
\usepackage{fancyref}
\usepackage{textcomp}

\usepackage{multirow}
\usepackage{stmaryrd}

\newenvironment{textbmatrix}{	\setlength{\arraycolsep}{2.5pt}%
								\big[\begin{matrix}}{\end{matrix}\big]%
								\raisebox{0.08ex}{\vphantom{M}}}

\def\be{\begin{equation}}
\def\ee{\end{equation}}
\def\een{\nonumber \end{equation}}
\def\mat{\begin{bmatrix}}
\def\emat{\end{bmatrix}}
\def\btm{\begin{textbmatrix}}
\def\etm{\end{textbmatrix}}

\def\ba#1\ea{\begin{align}#1\end{align}}
\def\bas#1\eas{\begin{align*}#1\end{align*}}
\def\bs#1\es{\begin{split}#1\end{split}} 
\def\bg#1\eg{\begin{gather}#1\end{gather}}
\def\bml#1\eml{\begin{multline}#1\end{multline}}
\def\bi#1\ei{\begin{itemize}#1\end{itemize}}

\newcommand{\lefto}{\mathopen{}\left}

\DeclareMathOperator{\tr}{tr}				
\DeclareMathOperator{\diag}{diag}			
\DeclareMathOperator*{\argmin}{arg\;min}		
\DeclareMathOperator{\kron}{\otimes}			
\DeclareMathOperator{\Exop}{\opE}			

\newcommand{\Ex}[2]{\ensuremath{\Exop_{#1}\lefto[#2\right]}} 	

\safemath{\dirac}{\delta}					
\safemath{\krond}{\dirac}					

\safemath{\upto}{\uparrow}
\safemath{\downto}{\downarrow}
\safemath{\iu}{j}							
\safemath{\ev}{\lambda}						
\safemath{\hilseqspace}{l^{2}}				
\newcommand{\banachfunspace}[1]{\setL^{#1}}	
\safemath{\hilfunspace}{\banachfunspace{2}}	

\safemath{\SNR}{\textsf{SNR}} 				
\safemath{\PAR}{\textsf{PAR}} 				
\safemath{\No}{N_0}							
\safemath{\Es}{E_s}							
\safemath{\Eb}{E_b}							
\safemath{\EbNo}{\frac{\Eb}{\No}}
\safemath{\EsNo}{\frac{\Es}{\No}}

\DeclareMathOperator{\CHop}{\ensuremath{\opH}} 
\safemath{\tvir}{\rndh_{\CHop}}				
\safemath{\tvtf}{\rndl_{\CHop}}				
\safemath{\spf}{\rnds_{\CHop}}				
\safemath{\bff}{H_{\CHop}}					

\safemath{\ircf}{r_{h}}						
\safemath{\tftvcf}{r_{s}}					
\safemath{\tfcf}{r_{l}}						
\safemath{\bfcf}{r_{H}}						

\safemath{\tcorr}{c_h}						
\safemath{\scf}{c_{s}}						
\safemath{\tfcorr}{c_{l}}					
\safemath{\fcorr}{c_{H}}						

\safemath{\mi}{I}							
\safemath{\capacity}{C}						

\safemath{\normal}{\mathcal{N}}			
\safemath{\jpg}{\mathcal{CN}}			
\safemath{\mchain}{\leftrightarrow}		

\safemath{\dB}{\,\mathrm{dB}}
\safemath{\dBm}{\,\mathrm{dBm}}
\safemath{\Hz}{\,\mathrm{Hz}}
\safemath{\kHz}{\,\mathrm{kHz}}
\safemath{\MHz}{\,\mathrm{MHz}}
\safemath{\GHz}{\,\mathrm{GHz}}
\safemath{\s}{\,\mathrm{s}}
\safemath{\ms}{\,\mathrm{ms}}
\safemath{\mus}{\,\mathrm{\text{\textmu}s}}
\safemath{\ns}{\,\mathrm{ns}}
\safemath{\ps}{\,\mathrm{ps}}
\safemath{\meter}{\,\mathrm{m}}
\safemath{\mm}{\,\mathrm{mm}}
\safemath{\cm}{\,\mathrm{cm}}
\safemath{\m}{\,\mathrm{m}}
\safemath{\W}{\,\mathrm{W}}
\safemath{\mW}{\, \mathrm{mW}}
\safemath{\J}{\,\mathrm{J}}
\safemath{\K}{\,\mathrm{K}}
\safemath{\bit}{\,\mathrm{bit}}
\safemath{\nat}{\,\mathrm{nat}}

\safemath{\define}{\triangleq}			

\safemath{\equivalent}{\sim}
\safemath{\distas}{\sim}					
\safemath{\sdiff}{\Delta}				

\safemath{\reals}{\mathbb{R}}
\safemath{\positivereals}{\reals_{+}}
\safemath{\integers}{\mathbb{Z}}
\safemath{\posint}{\integers_{+}}
\safemath{\naturals}{\mathbb{N}}
\safemath{\posnaturals}{\naturals_{+}}
\safemath{\complexset}{\mathbb{C}}
\safemath{\rationals}{\mathbb{Q}}

\newcommand*{\fancyrefapplabelprefix}{app}		
\newcommand*{\fancyrefthmlabelprefix}{thm}		
\newcommand*{\fancyreflemlabelprefix}{lem}		
\newcommand*{\fancyrefcorlabelprefix}{cor}		
\newcommand*{\fancyrefdeflabelprefix}{def}		
\newcommand*{\fancyrefproplabelprefix}{prop}	
\newcommand*{\fancyrefobslabelprefix}{obs}		
\newcommand*{\fancyrefalglabelprefix}{alg}		
\newcommand*{\fancyrefasmlabelprefix}{asm}	    
\newcommand*{\fancyrefasmslabelprefix}{asms}	    
\newcommand*{\fancyreftbllabelprefix}{tbl}	    
\newcommand*{\fancyrefestilabelprefix}{esti}	    

\frefformat{vario}{\fancyrefseclabelprefix}{Section~#1}
\frefformat{vario}{\fancyrefthmlabelprefix}{Theorem~#1}
\frefformat{vario}{\fancyreflemlabelprefix}{Lemma~#1}
\frefformat{vario}{\fancyrefcorlabelprefix}{Corollary~#1}
\frefformat{vario}{\fancyrefdeflabelprefix}{Definition~#1}
\frefformat{vario}{\fancyrefobslabelprefix}{Observation~#1}
\frefformat{vario}{\fancyrefasmlabelprefix}{Assumption~#1}
\frefformat{vario}{\fancyrefasmslabelprefix}{Assumptions~#1}
\frefformat{vario}{\fancyreffiglabelprefix}{Figure~#1}
\frefformat{vario}{\fancyrefapplabelprefix}{Appendix~#1} 
\frefformat{vario}{\fancyrefproplabelprefix}{Proposition~#1}
\frefformat{vario}{\fancyrefalglabelprefix}{Algorithm~#1}
\frefformat{vario}{\fancyrefeqlabelprefix}{(#1)}
\frefformat{vario}{\fancyreftbllabelprefix}{Table~#1}
\frefformat{vario}{\fancyrefestilabelprefix}{Estimator~#1}

%% file: defs.tex
\newtheorem{thm}{Theorem}
\newtheorem{cor}{Corollary}   
\newtheorem{prop}{Proposition}

\newtheorem{lem}{Lemma} 

\newtheorem{rem}{Remark}

\newtheorem{asms}{Assumptions}
\newtheorem{esti}{Estimator}

\safemath{\dictab}{[\,\dicta\,\,\dictb\,]}

\safemath{\ysig}{\bmy}
\safemath{\ysighat}{\hat{\ysig}}
\safemath{\ysigdim}{M}
\safemath{\xsig}{\bmx}
\safemath{\xsigdim}{N}
\safemath{\nx}{n_x}
\safemath{\zsig}{\bmz}
\safemath{\zsigdim}{\ysigdim}
\safemath{\rsig}{\bmr}
\safemath{\Adict}{\bA}
\safemath{\Adicttilde}{\widetilde{\Adict}}
\safemath{\Adictdim}{\outputdim\times\xsigdim}
\safemath{\avec}{\bma}
\safemath{\avectilde}{\tilde{\avec}}
\safemath{\Bdict}{\bB}
\safemath{\Bdicttilde}{\widetilde{\Bdict}}
\safemath{\Cdict}{\bC}
\safemath{\cvec}{\bmc}
\safemath{\Ddict}{\bD}
\safemath{\Ddictdim}{\ysigdim\times\xsigdim}
\safemath{\dvec}{\bmd}
\safemath{\Ddicttilde}{\widetilde{\bD}}
\safemath{\Bonb}{\bB}
\safemath{\bvec}{\bmb}
\safemath{\Bonbdim}{\ysigdim\times\ysigdim}
\safemath{\noise}{\bmn}
\safemath{\noisedim}{\ysigim}
\safemath{\err}{\bme}
\safemath{\errdim}{\ysigdim}
\safemath{\errset}{\setE}
\safemath{\nerr}{n_e}
\safemath{\delop}{\bP_\errset}
\safemath{\delopc}{\bP_{{\errset}^c}}

\safemath{\cplxi}{\imath}
\safemath{\cplxj}{\jmath}

\safemath{\dict}{\matD}
\safemath{\inputdim}{N}		
\safemath{\outputdim}{M}		
\safemath{\sparsity}{S}	
\safemath{\inputdimA}{{N_a}}	
\safemath{\inputdimB}{{N_b}}	
\safemath{\elemA}{{n_a}}	
\safemath{\elemB}{{n_b}}	
\safemath{\resA}{\matR_a}	
\safemath{\resB}{\matR_b}	
\safemath{\subD}{\matS} 
\safemath{\subA}{\matS_a} 
\safemath{\subB}{\matS_b} 
\safemath{\dicta}{\matA} 	
\safemath{\dictb}{\matB} 	
\safemath{\hollowS}{H}
\safemath{\hollowA}{H_a}
\safemath{\hollowB}{H_b}
\safemath{\cross}{Z}
\safemath{\coh}{\mu_d}			
\safemath{\coha}{\mu_a}			
\safemath{\cohb}{\mu_b}			
\safemath{\mubs}{\nu}	
\safemath{\cohm}{\mu_m} 
\safemath{\dictset}{\setD}	
\safemath{\dictsetp}{\dictset(\coh,\coha,\cohb)}	
\safemath{\dictsetgen}{\dictset_\text{gen}}
\safemath{\dictsetgenp}{\dictsetgen(\coh)}
\safemath{\dictsetonb}{\dictset_\text{onb}}
\safemath{\dictsetonbp}{\dictsetonb(\coh)}

\safemath{\leftside}{U}
\safemath{\rightsideA}{R_a}
\safemath{\rightsideB}{R_b}

\safemath{\indexS}{\setI_S} 

\safemath{\na}{n_a}			
\safemath{\nb}{n_b}			
\safemath{\coeffa}{p_i}	
\safemath{\coeffb}{q_j}	
\safemath{\seta}{\setP}		
\safemath{\setb}{\setQ}     
\safemath{\setw}{\setW}	
\safemath{\setz}{\setZ}	
\safemath{\cola}{\veca}		
\safemath{\colb}{\vecb}		
\safemath{\cold}{\vecd}		
\safemath{\inputvec}{\vecx} 	
\safemath{\error}{\vece}	
\safemath{\noiseout}{\vecz} 	
\safemath{\inputvecel}{x}
\safemath{\inputveca}{\vecx_a}
\safemath{\inputvecb}{\vecx_b}
\safemath{\outputvec}{\vecy}	
\safemath{\lambdamin}{\lambda_{\mathrm{min}}}

\safemath{\elltwo}{\ell_2}
\safemath{\ellone}{\ell_1}
\safemath{\ellzero}{\ell_0}
\safemath{\ellinf}{\ell_\infty}
\safemath{\ellinftilde}{\ell_{\widetilde\infty}}
\safemath{\licard}{Z(\coh,\coha,\cohb)}
\safemath{\xsol}{\hat{x}}
\safemath{\xbord}{x_b}		
\safemath{\xstat}{x_s}		
\safemath{\xstatLone}{\tilde{x}_s}
\safemath{\order}{\mathcal{O}} 
\safemath{\scales}{\Theta} 
\safemath{\ones}{\mathbf{1}} 
\safemath{\zeroes}{\mathbf{0}} 
\safemath{\thlone}{\kappa(\coh,\cohb)} 
\safemath{\constoneA}{\delta} 
\safemath{\constoneB}{\epsilon} 
\safemath{\nlarge}{L}				   
\safemath{\sumlarge}{S_\nlarge}
\safemath{\maxlarger}{P_\nlarge}	   
\safemath{\Pzero}{\textrm{P0}}	
\safemath{\Pone}{\textrm{P1}}
\safemath{\vecfir}{\vecw}			 
\safemath{\vecsec}{\vecz}
\safemath{\elvecfir}{w}              
\safemath{\elvecsec}{z}				 
\safemath{\nlargefir}{n}
\safemath{\normout}{\gamma}
\safemath{\auxfun}{h}
\safemath{\supp}{\textrm{supp}}

\safemath{\indexa}{\ell}
\safemath{\indexb}{r}
\safemath{\indexc}{i}
\safemath{\indexd}{j}

\safemath{\project}{P}

%% file: 1-introduction.tex

\section{Introduction}
\label{sec:introduction}
Phase retrieval refers to the problem of recovering an unknown $N$-dimensional signal vector $\bmx\in\setH^N$, with $\setH$ being the set of either real ($\reals$) or complex ($\complexset$) numbers, from the following nonlinear measurement process:  
\begin{align} \label{eq:measurementmodel}
\bmy = f(\bA\bmx+\bme^z)+\bme^y.
\end{align}
Here, the measurement vector $\bmy\in\reals^M$ contains $M$ real-valued observations, for example measured through the nonlinear function $f(z)=|z|^2$ that operates element-wise on vectors, $\bA\in\setH^{M\times N}$ is a given measurement matrix, and the vectors $\bme^z\in\setH^M$ and $\bme^y\in\reals^N$ model signal and measurement noises, respectively.
In contrast to the majority of existing results on phase retrieval that assume randomness in the measurement matrix $\bA$, we focus on the practical scenario in which the measurement matrix~$\bA$ is deterministic, but the signal vector $\bmx$ to be recovered as well as the two noise sources $\bme^z$ and $\bme^y$ are random.

\subsection{Phase Retrieval}
Phase retrieval has been studied extensively over the last decades~\cite{gerchberg1972practical,fienup1982phase} and finds use in a  range of applications, including imaging \cite{fogel2013phase,yeh2015experimental,holloway2016}, microscopy \cite{kou2010transport,faulkner2004movable}, and X-ray crystallography \cite{harrison1993phase,miao2008extending,pfeiffer2006phase}.
Phase retrieval problems were solved traditionally using alternating projection methods, such as the Gerchberg-Saxton~\cite{gerchberg1972practical} and Fienup~\cite{fienup1982phase} algorithms. More recent results have shown that semidefinite programming enables the design of algorithms with  performance guarantees~\cite{candes2013phaselift,candes2014solving,candes2015phase,waldspurger2015phase}. These methods lift the problem to a higher dimension, resulting in excessive complexity and memory requirements.
To perform phase retrieval for high-dimensional problems with performance guarantees, a range of convex  \cite{bahmani2017phase,phasemax,hand2016elementary,dhifallah2017phase,dhifallah2017fundamental,yuan2017phase,SAH18} and nonconvex methods~\cite{netrapalli2013phase,schniter2015compressive,candes2015wirtinger,chen2015solving,zhang2016reshaped,wang2016solving,zhang2016provable,wei2015solving,sun2016geometric,zeng2017coordinate,lu2017phase,ma2018optimization} have been proposed recently.

\subsection{Spectral Initializers}
All of the above non-lifting-based  phase retrieval methods rely on accurate initial estimates of the signal vector to be recovered. Such estimates are typically obtained by means of so-called \emph{spectral initializers} put forward in \cite{netrapalli2013phase}.
Spectral initializers  first compute a Hermitian matrix of the following form:
\begin{align} \label{eq:commonspectralinitializer}
\bD_\beta =   \beta\sum_{m=1}^M \setT(y_m)\bma_m\bma_m^H,
\end{align}
where $\beta>0$ is a suitably-chosen scaling factor, $y_m$ denotes the $m$th measurement, $\bma_m^H$ corresponds to the $m$th row of the measurement matrix~$\bA$ and $\setT: \reals\to\reals$ is a (possibly nonlinear) \emph{preprocessing function}. While the identity $\setT(y)=y$ was used originally in \cite{netrapalli2013phase}, recent results revealed that carefully crafted preprocessing functions yield more accurate estimates~\cite{chen2015solving,chen2015phase,wang2016solving,wang2017solving,lu2017phase,mondelli2017fundamental}. 
From the matrix~$\bD_\beta$ in \fref{eq:commonspectralinitializer}, one then extracts the (scaled) eigenvector~$\hat\bmx$ associated with the largest eigenvalue, which serves as an initial estimate of the solution to the phase retrieval problem.

As shown in~\cite{netrapalli2013phase,chen2015solving,chen2015phase,wang2016solving,wang2017solving,lu2017phase,mondelli2017fundamental},  
for i.i.d.\ Gaussian measurement matrices~$\bA$, sufficiently large measurement ratios $\delta=M/N$, and carefully crafted preprocessing functions~$\setT$, spectral initializers provide accurate initialization vectors.  
In fact, the results in \cite{mondelli2017fundamental} for the large-system limit with  $\delta$ fixed and $M\to\infty$ show that spectral initializers in combination with an optimal preprocessing function $\setT$ achieve the fundamental information-theoretic limits of phase retrieval. 
However, the assumption of having i.i.d.\ Gaussian measurement matrices~$\bA$ is impractical---it is more natural to assume that the signal vector $\bmx$ is random and the measurement matrix $\bA$ is deterministic and structured \cite{bendory2016non}.

\subsection{Contributions}
We propose a novel class of estimators, called \emph{linear spectral estimators (LSPEs)}, that provide accurate estimates for general nonlinear measurement systems of the form~\fref{eq:measurementmodel} and enable a nonasymptotic mean-squared error (MSE) analysis. 
We showcase the efficacy of LSPEs by applying them to phase retrieval problems, where we compute initialization vectors for real- and complex-valued systems with deterministic and finite-dimensional measurement matrices.
For the proposed LSPEs, we derive nonasymptotic and sharp bounds on the MSE for signal estimation from phaseless measurements.
We use synthetic and real-world phase retrieval problems to demonstrate that LSPEs are able to significantly outperform existing spectral initializers on systems that acquire structured measurements.
We furthermore show that preprocessing the phaseless measurements enables LSPEs to generate improved initialization vectors for an even broader class of measurement systems.

\subsection{Notation}
Lowercase and uppercase boldface letters represent column vectors and matrices, respectively. For a matrix $\bA$, its transpose and Hermitian conjugate is $\bA^T$ and $\bA^H$, respectively, and the $k$th row and $\ell$th column entry is~$[\bA]_{k,\ell}=A_{k,\ell}$. For a vector~$\bma$, the $k$th entry is~$[\bma]_k=a_k$.
The $\ell_2$-norm of~$\bma$ is denoted by~$\|\bma\|_2$ and the Frobenius norm of~$\bA$ by $\|\bA\|_F$. The Kronecker product is~$\kron$, the  Hadamard product is $\odot$, the  Hadamard division is $\oslash$,  and the trace operator is $\tr(\cdot)$. The $N\times N$ identity matrix is denoted by $\bI_{N}$; the $M\times N$ all-zeros and all-ones matrices are denoted by $\bZero_{M\times N}$ and $\bOne_{M\times N}$, respectively. 
For a vector $\bma$, $\diag(\bma)$ is a square matrix with $\bma$ on the main diagonal; for a matrix $\bA$, $\diag(\bA)$ is a column vector containing the diagonal elements of~$\bA$. 
%
%


%% file: 2-estimator.tex

\section{Linear Spectral Estimators}
\label{sec:LMSPE}
We start by reviewing the essentials of spectral initializers and then, introduce linear spectral estimators (LSPEs) for  measurement systems of the form \fref{eq:measurementmodel} with general nonlinearities $f$. We furthermore provide nonasymptotic expressions for the associated estimation error, and we compare our analytical results to that of conventional spectral initializers in \fref{eq:commonspectralinitializer}. In \fref{sec:phaseretrieval}, we will apply LSPEs to phase retrieval.

\subsection{Spectral Estimation and Initializers}
One of the key issues of the phase retrieval problem is the fact that if~$\bmx$ is a solution to \fref{eq:measurementmodel}, then~$e^{j\phi}\bmx$  for any $\phi\in[0,2\pi)$ is also a valid solution (assuming $\setH=\complexset$). Put simply, the solution is nonunique up to a global phase shift. 
One way of combating this issue is to directly recover the outer product~$\bmx\bmx^{H}$ instead of $\bmx$, which is unaffected by phase shifts; this insight is the key underlying lifting-based phase retrieval methods~\cite{candes2013phaselift,candes2014solving,candes2015phase,waldspurger2015phase}.
With this in mind, one could envision the design of an estimator that directly minimizes the conditional MSE:
\begin{align} \label{eq:PME}
\dot\bmx = \argmin_{\tilde{\bmx}\in\setH^{N}}  \Ex{}{\|\bmx\bmx^{H}-\tilde{\bmx}\tilde{\bmx}^H\|_F^2\mid\bmy}\!.
\end{align}
Here, expectation is with respect to the signal vector~$\bmx$ and the two noise sources $\bme^z$ and $\bme^y$.  This optimization problem resembles that of a posterior mean estimator (PME) which is, in general, difficult to derive, even for simple observation models---for phase retrieval, we have two additional challenges: (i) nonlinear phaseless measurements as in~\fref{eq:measurementmodel} and (ii) the quantity $\tilde\bmx\tilde\bmx^H$ has rank-1.

Spectral initializers avoid the issues of the estimator in \fref{eq:PME} by first replacing the true outer product $\bmx\bmx^H$ with a so-called spectral estimator matrix $\bD_\beta$ as in~\fref{eq:commonspectralinitializer} that depends on the measurement vector $\bmy$. In a second step, one then computes the best rank-1 approximation as follows:
\begin{align} \label{eq:spectralestimatormatrix}
\hat\bmx  =\argmin_{\tilde{\bmx}\in\setH^{N}} \|\bD_\beta-\tilde\bmx\tilde\bmx^{H}\|_F^2 
\end{align}
from which the estimate $\hat\bmx$ can be extracted.
By performing an eigenvalue decomposition $\bD_\beta=\bU\bLambda\bU^H$ with $\bU^H\bU=\bI_M$ and the eigenvalues in the diagonal matrix $\bLambda=\mathrm{diag}([\lambda_1,\ldots,\lambda_M]^T)$ are sorted in descending order of their magnitudes, a spectral initializer is given by the \emph{scaled leading eigenvector} $\hat\bmx=\sqrt{\lambda_1}\bmu_1$.
In practice, one can use power iterations to efficiently compute  $\hat\bmx$.

\subsection{Linear Spectral Estimators}
We now propose a novel class of estimators, which we call \emph{linear spectral estimators (LSPEs)}, that provide accurate estimates for general nonlinear measurement systems of the form~\fref{eq:measurementmodel}. To this end, we borrow ideas from the spectral initializer, the PME in \fref{eq:PME}, and the linear phase retrieval algorithm put forward in \cite{ghods2018phaselin}.
In the first step, LSPEs apply a linear estimator to the nonlinear observations in~$\setT(\bmy)$ to construct a \emph{spectral estimator matrix} $\bD_\bmy$ for which the \emph{spectral MSE} (or matrix MSE) defined as
\begin{align} \label{eq:MSE}
\textit{S-MSE} = \Ex{}{\left\| \bD_\bmy -\bmx\bmx^{H}\right\|_F^2}
\end{align}
is minimal.  We restrict ourselves to spectral estimator matrices $\bD_\bmy$ that are \emph{affine} in~$\setT(\bmy)$, i.e., are of the form
\begin{align} \label{eq:ansatz}
\bD_\bmy = \bW_0 + \sum_{m=1}^M \setT(y_m) \bW_m
\end{align}
with $\bW_m\in\setH^{N\times N}$, $m=0,\ldots,M$.
In the second step, we use the spectral estimator matrix $\bD_\bmy$ to extract a (scaled) leading eigenvector as  in \fref{eq:PME}, which is the \emph{linear spectral estimate} of the signal vector~$\bmx$.
Intuitively, if we can construct a matrix $\bD_\bmy$ from the preprocessed measurements in $\setT(\bmy)$ for which the S-MSE in~\fref{eq:MSE} is minimal, then we expect that computing its best rank-1 approximation would yield an accurate estimate of the signal vector~$\bmx$ up to a global phase shift. We will justify this claim in \fref{sec:MSEandAngularError}.

Mathematically, we wish to compute a matrix $\bD_\bmy$ of the form \fref{eq:ansatz} that is the solution to the following  problem:
\begin{align} \label{eq:optimalestimator}
\underset{\begin{subarray}{c}
\widetilde\bW_{m}\in\setH^{N\times N} \\ m=0,\ldots,M
\end{subarray}}{\mathrm{minimize}} \Ex{}{\left\| \widetilde{\bW}_0 +\! \sum_{m=1}^M \setT(y_m) \widetilde\bW_m\!-\bmx\bmx^{H}\right\|_F^2}\!\!.
\end{align}
Clearly,  the spectral estimator matrix $\bD_\bmy$ will depend on the measurement matrix $\bA$, the statistics of the signal to be estimated $\bmx$ and the two noise sources $\bme^z$ and $\bme^y$, the nonlinearity~$f$, as well as the preprocessing function~$\setT$. 
For this setting, we have the following general result which summarizes the LSPE; the proof is given in \fref{app:LMSPE}.
\begin{thm}[Linear Spectral Estimator] \label{thm:LMSPE}
Let the measurement vector $\bmy$ be a result of the general measurement model in \fref{eq:measurementmodel} and select a preprocessing function $\setT$.
Define the vector $\overline\setT(\bmy)= \Ex{}{\setT(\bmy)}$ and assume the matrix
\begin{align*}
\bT = \Ex{}{(\setT(\bmy)-\overline\setT(\bmy))(\setT(\bmy)-\overline\setT(\bmy))^T}
\end{align*}
is full rank. Let $\bmt\in\reals^M$ satisfy $\bT \bmt  = \setT(\bmy)-\overline\setT(\bmy)$  and 
\begin{align*}
\bV_m = \Ex{}{(\setT(y_m)-\overline\setT(y_m))(\bmx\bmx^{H}-\bK_{\bmx})}
\end{align*}
for $m=1,\ldots,M$ with $\bK_{\bmx}=\Ex{}{\bmx\bmx^{H}}$. 
Then, the LSPE matrix that minimizes the S-MSE in \fref{eq:MSE} is given by
\begin{align} \label{eq:LMSPEmatrix}
\bD_\bmy = \bK_{\bmx} + \sum_{m=1}^M t_m  \bV_m.
\end{align}
The linear spectral estimate $\hat\bmx$ is then given by the scaled leading eigenvector of the matrix $\bD_\bmy$ in \fref{eq:LMSPEmatrix}.
\end{thm}

The vector $\bmt$ is the only quantity in \fref{thm:LMSPE} that depends on the \emph{actual} (nonlinear) observations contained in the measurement vector $\bmy$. 
All other quantities depend only on the first two moments of~$\bmx\bmx^H$ as well as the considered signal, noise, and measurement models.
 The key features of the LSPE are as follows: (i) the involved quantities can often be computed in closed form (see \fref{sec:phaseretrieval} for two applications to phase retrieval) and (ii) LSPEs enable a nonasymptotic and sharp analysis of the associated estimation error.

\begin{rem}
\fref{thm:LMSPE} requires the matrix $\bT$ to be invertible. This condition is satisfied in most practical situations with nondegenerate measurement matrices~$\bA$ or in situations with nonzero measurement noise.
\end{rem}

\subsection{Estimation Error Analysis of LSPEs}
\label{sec:MSEandAngularError}

The remaining piece of the proposed LSPE is to show that the result of this two-step estimation procedure indeed yields a vector that is close to the signal vector $\bmx$. 
We start with the following result; the proof is given in \fref{app:MSEofLMSPE}.
\begin{thm}[S-MSE of the LSPE] \label{thm:MSEofLMSPE}
Let the assumptions of \fref{thm:LMSPE} hold.
Then, the S-MSE in \fref{eq:MSE} for the LSPE matrix in~\fref{eq:LMSPEmatrix} is given by  
\begin{align} \label{eq:MSEofLMSPE}
\textit{S-MSE}_\text{LSPE} = C_{\bmx\bmx^{H}} \!-\!\!  \sum_{m=1}^M \sum_{m'=1}^M [ \bT^{-1}]_{m,m'} \tr\!\left( \bV^H_m \bV_{m'} \right) 
\end{align}
with $C_{\bmx\bmx^{H}}=\Ex{}{\left\|  \bmx\bmx^{H}-\bK_{\bmx}\right\|_F^2}$.
\end{thm}

With this result, we are ready to establish a bound on the estimation error of the LSPE. The proof of the following result follows from \fref{thm:MSEofLMSPE} and is given in \fref{app:LMSPEapproxerror}.

\begin{cor}[LSPE Estimation Error] \label{cor:LMSPEapproxerror}
Let the assumptions of \fref{thm:LMSPE} hold. Then, the estimation error (EER) of the LSPE satisfies the following inequality:
\begin{align} \label{eq:estimationerror}
\textit{EER}_\text{LSPE} = \Ex{}{ \| \hat\bmx\hat\bmx^H- \bmx\bmx^{H} \|_F^2 } \leq 4\, \textit{S-MSE}_\text{LSPE}.
\end{align}
\end{cor}

This result implies that by minimizing the S-MSE in \fref{eq:MSE} via~\fref{eq:optimalestimator}, we are also reducing the EER of the LSPE. In other words, if the spectral error $\bE = \bD_\bmy - \hat\bmx\hat\bmx^H$ is small, then the EER of the LSPE~\fref{eq:estimationerror} will be small.

\begin{rem}
\fref{cor:LMSPEapproxerror} is nonasymptotic and depends on the instance of measurement matrix $\bA$. This result is in stark contrast to existing performance bounds for spectral initializers~\cite{netrapalli2013phase,chen2015solving,chen2015phase,wang2016solving,wang2017solving} that strongly rely on randomness in the measurement matrix. In addition to randomness, the sharp performance guarantees in \cite{lu2017phase,mondelli2017fundamental} focus on the asymptotic regime for which $\delta=M/N$ is fixed and $M\to\infty$.
\end{rem}

\subsection{S-MSE of Spectral Initializers}
We can also derive an exact expression for the S-MSE of the conventional spectral initializer in~\fref{eq:commonspectralinitializer}. We assume optimal scaling, i.e., the parameter $\beta$  is set to minimize the S-MSE. The following result characterizes the S-MSE of such a scaled spectral initializer; the proof is given in \fref{app:SIapproxerror}. 
\begin{prop}[S-MSE of the Spectral Initializer] \label{prop:SIapproxerror}
Let  $\bD_\beta$ be the conventional spectral initializer matrix in~\fref{eq:commonspectralinitializer}. Then, the optimally-scaled S-MSE defined as
\begin{align} \label{eq:MSEoptimallyscaled}
\textit{S-MSE}_\text{SI} = \min_{\beta\in\setH} \,\Ex{}{ \| \bD_\beta - \bmx\bmx^{H}\|_F^2 }
\end{align}
is given by
\begin{align}  \label{eq:MSESI}
\textit{S-MSE}_\text{SI} = R_{\bmx\bmx^H}\! - \frac{\left| \sum_{m=1}^M \bma_m^H \widetilde{\bV}_m \bma_m \right|^2}{ \sum_{m=1}^M\sum_{m'=1}^M \widetilde{T}_{m,m'}|\bma_m^H\bma_{m'}|^2},
\end{align}
where $R_{\bmx\bmx^H}=\Ex{}{\|\bmx\bmx^{H}\|_F^2}$, $\widetilde\bV_m=\Ex{}{\setT(y_m)\bmx\bmx^{H}}$, $m=1,\ldots,M$, and  $\widetilde\bT=\Ex{}{\setT(\bmy)\setT(\bmy)^T}$. 
\end{prop}

Since the matrix in~\fref{eq:commonspectralinitializer} is a special case of the LSPE matrix in \fref{eq:ansatz}, we have the following simple yet important property:
\begin{align*}
\textit{S-MSE}_\text{LSPE} \leq \textit{S-MSE}_\text{SI}.
\end{align*}
In words, the spectral MSE of the LSPE  cannot be worse than that of a spectral initializer. 
As we will show in \fref{sec:results}, LSPEs are able to outperform spectral initializers on both synthetic and real-world phase retrieval problems given that the same preprocessing function $\setT$ is used.

%% file: 3-phaseretrieval.tex

\section{LSPEs for Phase Retrieval Problems}
\label{sec:phaseretrieval}

The LSPE provides a framework for estimating signal vectors  from the general observation model in \fref{eq:measurementmodel}. To make the concept of LSPEs explicit and to demonstrate their efficacy in practice, we now show two application examples  to phase retrieval in complex-valued systems. The LSPE for real-valued phase retrieval can be found in \fref{app:realvaluedphase}.

\subsection{Phase Retrieval without Preprocessing}
We first focus on the case where the signal vector~$\bmx$ to be estimated and the measurement matrix $\bA$ are both complex-valued. The phaseless measurements $\bmy$, however, remain real-valued. We need the following assumptions.

\begin{asms} \label{asms:asms2}
Let $\setH=\complexset$. 
Assume square absolute measurements $f(z)=|z|^2$ and the identity preprocessing function $\setT(y)=y$. %
Assume that the signal vector $\bmx\in\complexset^N$ is i.i.d.\ circularly-symmetric complex Gaussian with covariance matrix $\bC_{\bmx} = \sigma_x^2 \bI_N$, i.e.,  $\bmx \sim \setC\setN(\bZero_{N\times1}, \sigma_x^2 \bI_N)$.
Assume that the signal noise vector $\bme^z$ is circularly-symmetric complex Gaussian with covariance matrix $\bC_{\bme^z}$, i.e., $\bme^z \sim \setC\setN(\bZero_{M\times 1},\bC_{\bme^z})$, and the measurement noise vector $\bme^y$ is a real-valued Gaussian vector with mean~$\bar\bme^y$ and covariance matrix $\bC_{\bme^y}$, i.e.,  $\bme^y \sim \setN(\bar\bme^y,\bC_{\bme^y})$. Furthermore assume that $\bmx$, $\bme^z$, and $\bme^y$ are independent. 
\end{asms}

Under these assumptions, we can derive the following LSPE which we call LSPE-$\complexset$; the detailed derivations of this spectral estimator are given in \fref{app:phaseinitC}.

\begin{esti}[LSPE-$\complexset$] \label{esti:phaseinitC}
Let \fref{asms:asms2} hold. Then, the spectral estimation matrix is given by
\begin{align} \label{eq:esti2}
\bD_\bmy^\complexset = \bK_\bmx+  \sum_{m=1}^M t_m \bV_m,
\end{align}
where $\bK_\bmx =\sigma_x^2 \bI_N $, 
the vector $\bmt\in\reals^M$ is given by the solution to the linear system $\bT \bmt  = \bmy-\overline\bmy$ with
\begin{align*}
\overline\bmy & =\diag(\bC_{\bmz})+\bar\bme^y \\
\bC_\bmz & = \sigma_x^2 \bA \bA^H+\bC_{\bme^z} \\
\bT &=  \bC_{\bmz} \odot \bC_{\bmz}^* + \bC_{\bme^y} 
\end{align*}
and $\bV_m=  \sigma_x^4 \bma_m  \bma_m^H $, $m=1,\ldots,M$.
The spectral estimate  $\hat\bmx$ is given by the (scaled) leading eigenvector of $\bD_\bmy^\complexset $ in \fref{eq:esti2}. Furthermore, the  S-MSE is given by \fref{thm:MSEofLMSPE}.
\end{esti} 

We emphasize that the spectral estimator matrix in \fref{eq:esti2} resembles that of the conventional spectral initializer matrix~\fref{eq:commonspectralinitializer} with the following key differences. First and foremost,  each outer product contained in $\bV_m= \sigma_x^4 \bma_m  \bma_m^H$ in \fref{esti:phaseinitC} is weighted by $t_m$, which is a function of \emph{all} phaseless measurements in~$\bmy$ and of the covariance matrix~$\bC_\bmx$. In contrast, each outer product in the conventional spectral initializer matrix in~\fref{eq:commonspectralinitializer} is only weighted by the associated measurement $y_m$. This difference enables the LSPE to weight each outer product depending on  correlations in the phaseless measurements caused by structure in the matrix~$\bA$. Second, the spectral estimator matrix includes a mean term~$\bK_\bmx$, which is absent in the spectral initializer matrix. 
As we will show in \fref{sec:results}, for the same preprocessing function $\setT$, \fref{esti:phaseinitC} is able to outperform spectral initializers for systems with structured measurement matrices~$\bA$. For large i.i.d.\ Gaussian measurement matrices, there is no particular correlation structure to exploit and LSPEs perform on par with spectral initializers.

%

\subsection{Phase Retrieval with Exponential Preprocessing}
To demonstrate the flexibility and generality of our framework, we now design an LSPE with an exponential preprocessing function for complex-valued phase retrieval. 
We derive the LSPE under the following assumptions.
\begin{asms} \label{asms:asms3}
Let $\setH=\complexset$. 
Assume square absolute measurements $f(z)=|z|^2$ and the exponential preprocessing function \mbox{$\setT(y)=\exp(-\gamma y )$} with $\gamma>0$, i.e., we consider
\begin{align*}
\setT(\bmy) = \exp\!\left(-\gamma(|\bmz|^2+\bme^y)\right) \quad \text{and} \quad \bmz=\bA\bmx+\bme^z,
\end{align*}
where the exponential function is applied element-wise to vectors.
The remaining assumptions are the same as in \fref{asms:asms2}.
\end{asms}

We now derive the following LSPE called LSPE-Exp; the derivation of this spectral estimator is given in \fref{app:phaseinitExp}.

\begin{esti}[LSPE-Exp] \label{esti:phaseinitExp}
Let \fref{asms:asms3} hold. Then, the spectral estimation matrix is given by
\begin{align} \label{eq:esti3}
\bD_\bmy^\text{\em Exp} = \bK_\bmx+  \sum_{m=1}^M t_m \bV_m,
\end{align}
where $\bK_\bmx =\sigma_x^2 \bI_N $, 
the vector $\bmt\in\reals^M$ is given by the solution to the linear system $\bT \bmt  = \setT(\bmy)-\overline\setT(\bmy)$ with
\begin{align*}
\overline\setT(\bmy) & =   \bmp_\gamma \oslash \bmq_\gamma \\
 \bT  & =   (\bmp_\gamma\bmp_\gamma^T) \odot \!\exp(\gamma^2\bC_{\bme^y}) \!\oslash\! ( \bmq_\gamma\bmq^T_\gamma\!  - \gamma^2 \bC_\bmz\odot \bC_\bmz^*) \\  
& \quad \, - (\bmp_\gamma\bmp_\gamma^T) \oslash   (\bmq_\gamma\bmq_\gamma^T) \\
\bV_m &= - \frac{\gamma\sigma_x^4[\bmp_\gamma]_m}{(\gamma[\bC_\bmz]_{m,m}+1)^2} \bma_m\bma^H_m, \quad m=1,\ldots,M,
\end{align*}
where  we use the following definitions:
\begin{align*}
\bmq_\gamma& =\gamma\diag(\bC_\bmz)+\bOne_{M\times1}  \\ 
\bmp_\gamma&=\exp\!\left(-\gamma \bar\bme^y + \gamma^2\textstyle\frac{1}{2}\diag(\bC_{\bme^y})\right)\\
\bC_\bmz & = \sigma_x^2 \bA \bA^H+\bC_{\bme^z}.
\end{align*}
The spectral estimate  $\hat\bmx$ is given by the (scaled) leading eigenvector of $\bD_\bmy^\text{\em Exp} $ in \fref{eq:esti3}. Furthermore, the  S-MSE of this estimator is given by \fref{thm:MSEofLMSPE}.
\end{esti} 

At first sight, the choice of the exponential preprocessing function used in \fref{esti:phaseinitExp} seems to be arbitrary. We emphasize, however, that this particular function is inspired by the asymptotically-optimal preprocessing function for properly-normalized Gaussian measurement ensembles proposed in \cite{mondelli2017fundamental} which is given by
\begin{align} \label{eq:optimalpreprocessor}
\setT_\text{opt}(y)=\frac{y-1}{y+\sqrt{\delta}-1}.
\end{align} 
As it turns out, we can scale, negate, and shift the exponential preprocessing function $\setT(y)=\exp(-\gamma y)$ to make it take a similar shape as the function in \fref{eq:optimalpreprocessor}.
More concretely, exponential preprocessing as well as $\setT_\text{opt}(y)$ enables one to attenuate the effect of measurements with large magnitude, which is also the idea underlying the class of orthogonal spectral initializers, as proposed in \cite{chen2015phase,wang2016solving,wang2017solving}, that perform well in practice.

%% file: 4-results.tex
\captionsetup[figure]{textfont=it,font=normalsize}
\begin{figure*}[tp]
\centering
\captionsetup[subfigure]{textfont=it,font=small}
\begin{subfigure}[b]{0.665\columnwidth}
\centering 
\includegraphics[width=1\columnwidth]{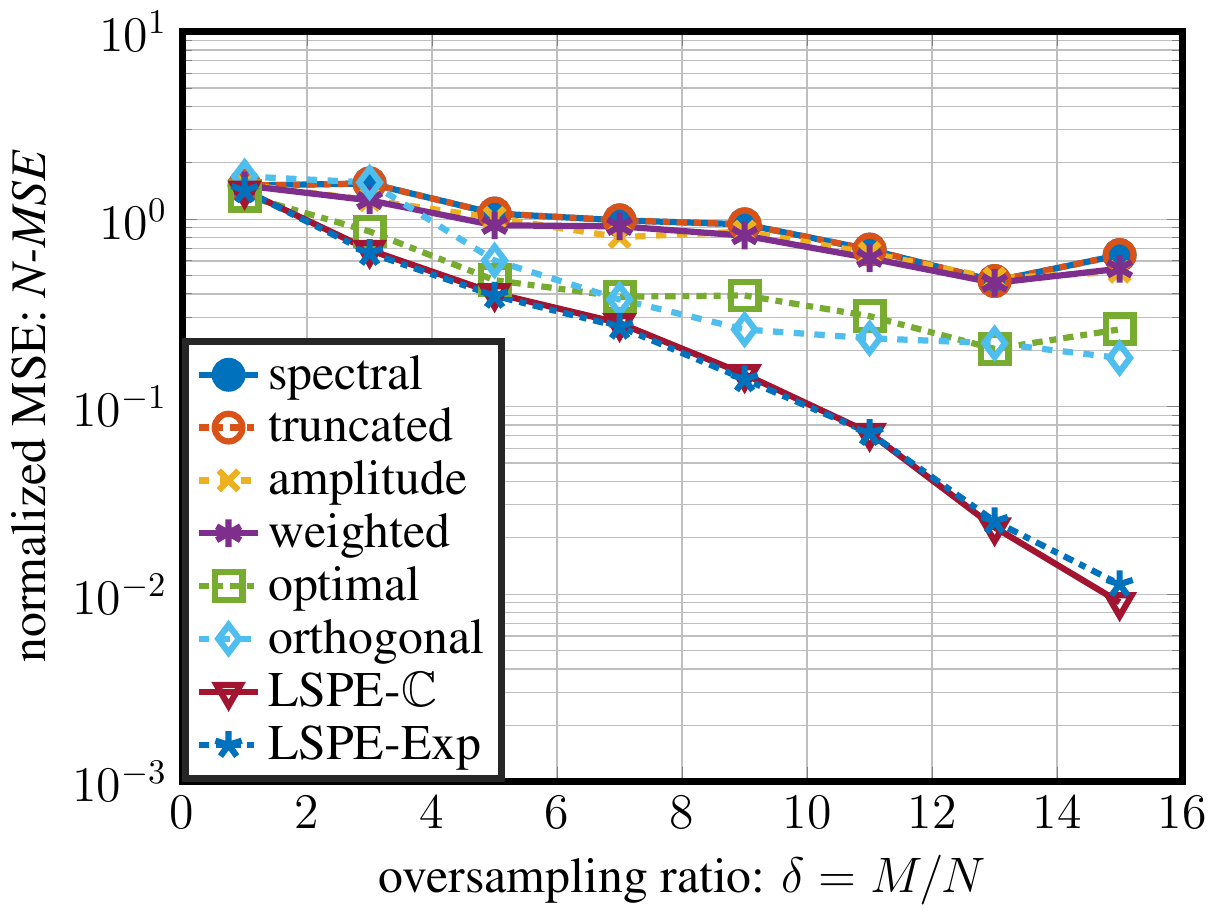}
\caption{Gaussian measurements, $N=16$.}
\label{fig:Gaussian16_MN}
\end{subfigure} 
\hfill
\begin{subfigure}[b]{0.665\columnwidth}
\centering 
\includegraphics[width=1\columnwidth]{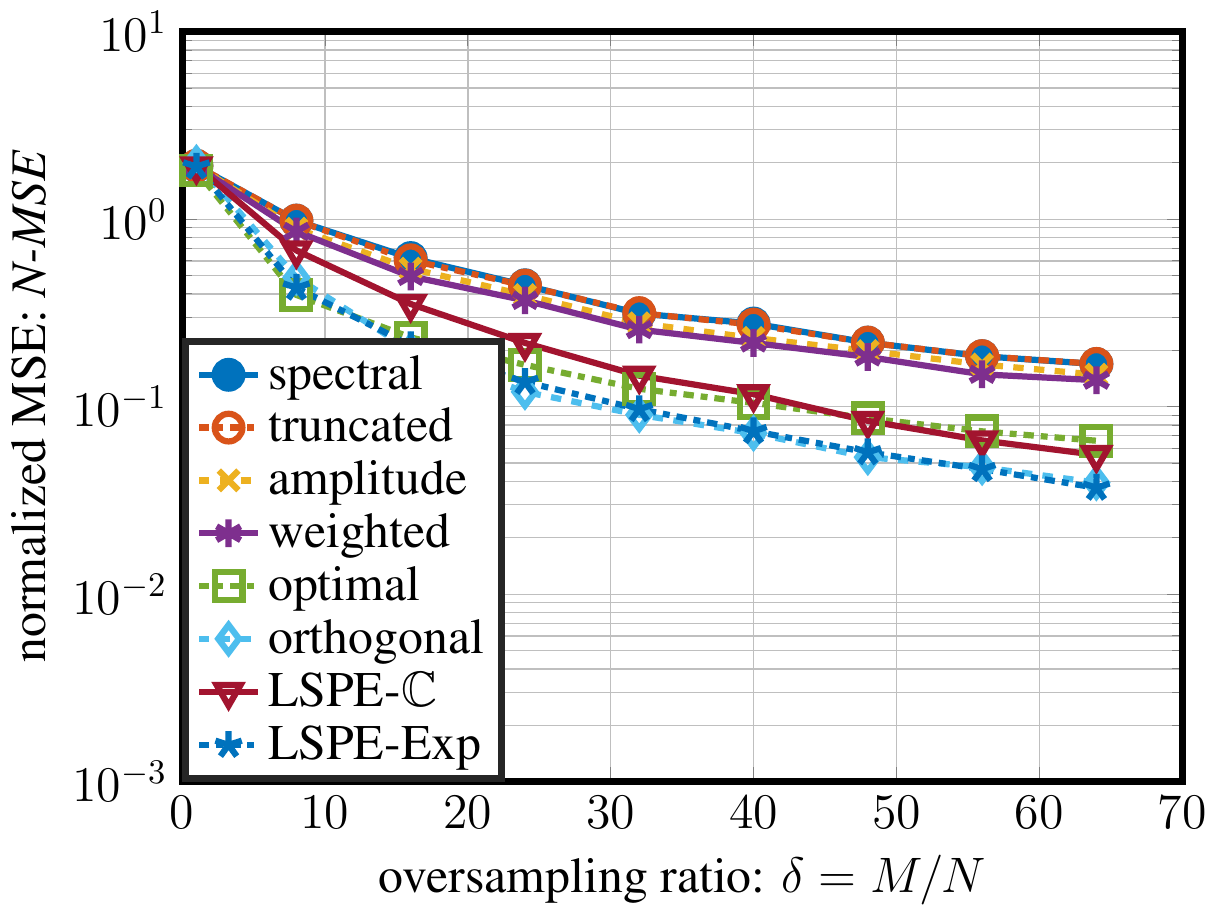}
\caption{Gaussian measurements, $N=256$.}
\label{fig:Gaussian256_MN}
\end{subfigure} 
\hfill
\begin{subfigure}[b]{0.665\columnwidth}
\centering 
\includegraphics[width=1\columnwidth]{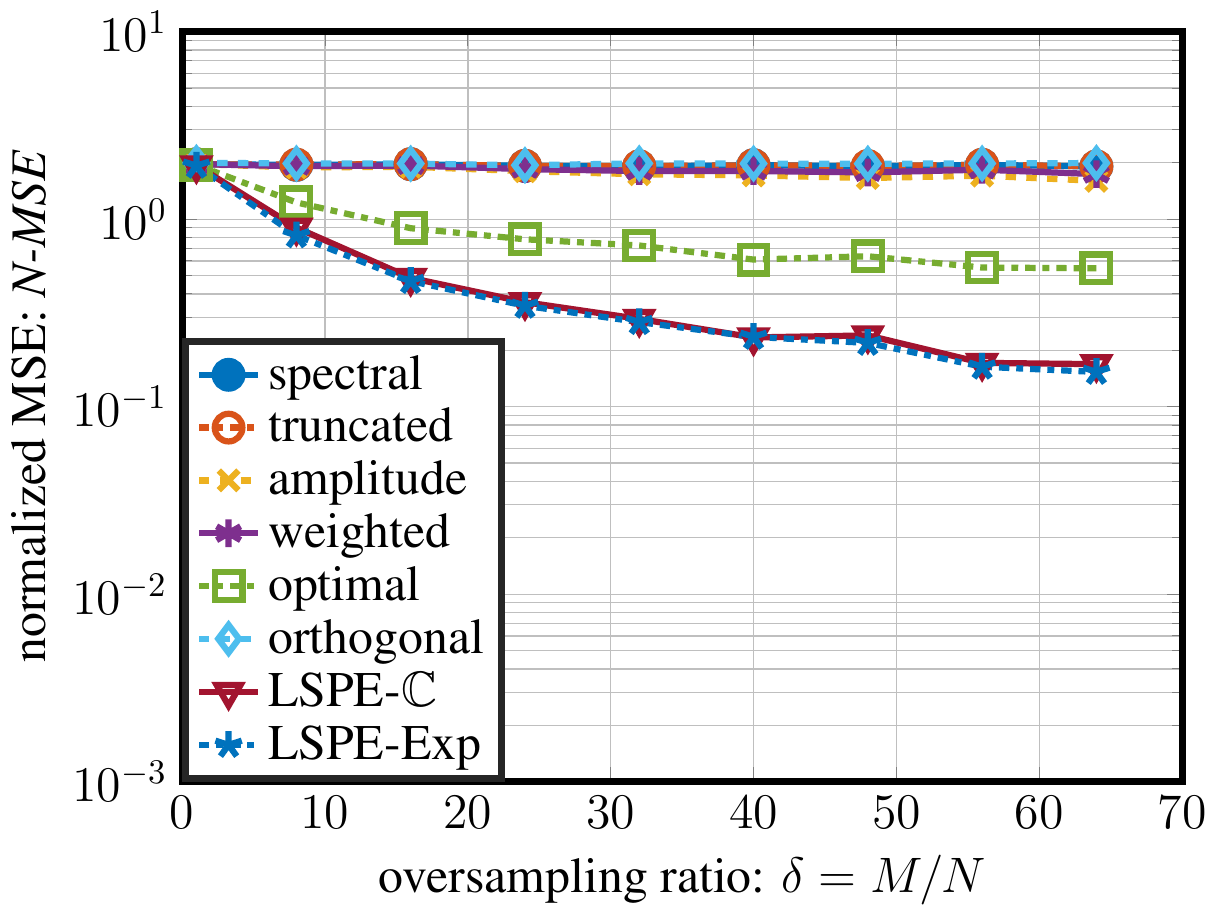}
\caption{Transmission measurements, $N=256$.}
\label{fig:TM256_MN}
\end{subfigure}
\caption{Comparison of normalized MSE (N-MSE) as a function of the oversampling ratio $\delta=M/N$ for complex-valued phase retrieval with different spectral initializers and with different measurement matrices. The proposed LSPEs perform well on low-dimensional problems, for structured measurement ensembles, or at high oversampling ratios~$\delta$.}
\vspace{-0.2cm}
\label{fig:comparison}
\end{figure*}

\section{Numerical Results}
\label{sec:results}
We now compare the performance of our LSPEs against existing spectral initializers proposed for phase retrieval on synthetic and real image data. All our results use the spectral initializers and experimental setups provided by PhasePack~\cite{chandra2017phasepack}.

\subsection{Impact of Measurement Ensemble}
\label{sec:resultsensemble}
We start by comparing the normalized MSE (N-MSE) defined as~\cite{chandra2017phasepack} 
\begin{align*}
\textit{N-MSE} = \frac{\min_{\alpha\in\setH} \|\bmx-\alpha\hat\bmx\|^2}{\|\bmx\|^2}
\end{align*}
for a range of spectral initializers  on different measurement ensembles. Specifically, we focus on the complex-valued case and consider (i) an i.i.d.\ Gaussian measurement matrix with signal dimension $N=16$, (ii) an i.i.d.\ Gaussian measurement matrix with $N=256$, and (iii) the structured ``transmission matrix'' used for image recovery through multiple scattering media as detailed in~\cite{metzler2017coherent}.
We vary the oversampling ratio $\delta=M/N$ and compare the N-MSE of the proposed complex-valued LSPEs, LSPE-$\complexset$ (\fref{esti:phaseinitC}) and LSPE-Exp (\fref{esti:phaseinitExp} with $\gamma=0.001$), to the following spectral initializers:
the original spectral initializer~\cite{netrapalli2013phase,candes2015phase} called ``spectral,''
truncated spectral initializer \cite{chen2015solving} called ``truncated,''
weighted spectral initializer \cite{wang2017solving} called ``weighted,''
amplitude spectral initializer \cite{wang2016solving} called ``amplitude,'' 
orthogonal spectral initializer \cite{chen2015phase} called ``orthogonal,'' 
and the asymptotically-optimal spectral initializer~\cite{mondelli2017fundamental} called ``optimal.''
For the following synthetic experiments, we generate the signals to be recovered according to \fref{asms:asms2} and \fref{asms:asms3} for \mbox{LSPE-$\complexset$} and LSPE-Exp, respectively. 


\fref{fig:Gaussian16_MN} shows that the proposed LSPEs significantly outperform all existing spectral initializers for small problem dimensions with Gaussian measurements; this improvement is even more pronounced for large oversampling ratios. 
The reason is that since we randomly generate a low-dimensional sensing matrix, the system will exhibit strong correlations among the measurements that can be exploited by LSPEs.
For larger dimensions with Gaussian measurements, we see in \fref{fig:Gaussian256_MN} that the proposed LSPEs do not provide an advantage over other methods. In fact, only LSPE-Exp  is able to perform as well as the orthogonal spectral initializer, which achieves the best performance in this scenario.
This behavior can be attributed to the facts that (i) for large random matrices there is no particular correlation structure among the measurements to exploit and (ii) ignoring measurements associated to large values in $y_m$ is increasingly important. 
For structured measurements, as it is the case for the transmission matrix from~\cite{metzler2017coherent}, we see in \fref{fig:TM256_MN} that LSPEs significantly outperform existing methods that are designed for random measurement ensembles.
In this scenario, exponential preprocessing does not improve performance since correlations in the transmission matrix are dominating the performance.

\subsection{S-MSE Expressions and Approximation Error}
We now validate our theoretical S-MSE expressions in \fref{thm:MSEofLMSPE} and \fref{prop:SIapproxerror}, and confirm the accuracy of the EER bound given in \fref{cor:LMSPEapproxerror}. 
In the following experiment, we set $M=8N$ and vary the dimension~$N$ from $8$ to $64$. 
For each pair $(M,N)$, we randomly generate one instance of an i.i.d.\ circularly symmetric complex Gaussian measurement matrix and average the different errors (S-MSE and EER) over $10,000$ Monte-Carlo trials. We consider a noiseless setting and  assume identity preprocessing, i.e., $\setT(y)=y$. The signal vectors are generated according to an i.i.d.\  circularly complex Gaussian random vector.
From \fref{fig:comparison}, we see that our analytical S-MSE expressions for the LSPE-$\complexset$ and spectral initializers match their empirical values. We furthermore see that the empirical EER is only about 6\,dB to 10\,dB lower than our non-asymptotic upper bound given in \fref{cor:LMSPEapproxerror}. 

\setlength{\textfloatsep}{15pt}
\captionsetup[figure]{textfont=it,font=normalsize}
\begin{figure}[t]
\centering
\includegraphics[width=0.8\columnwidth]{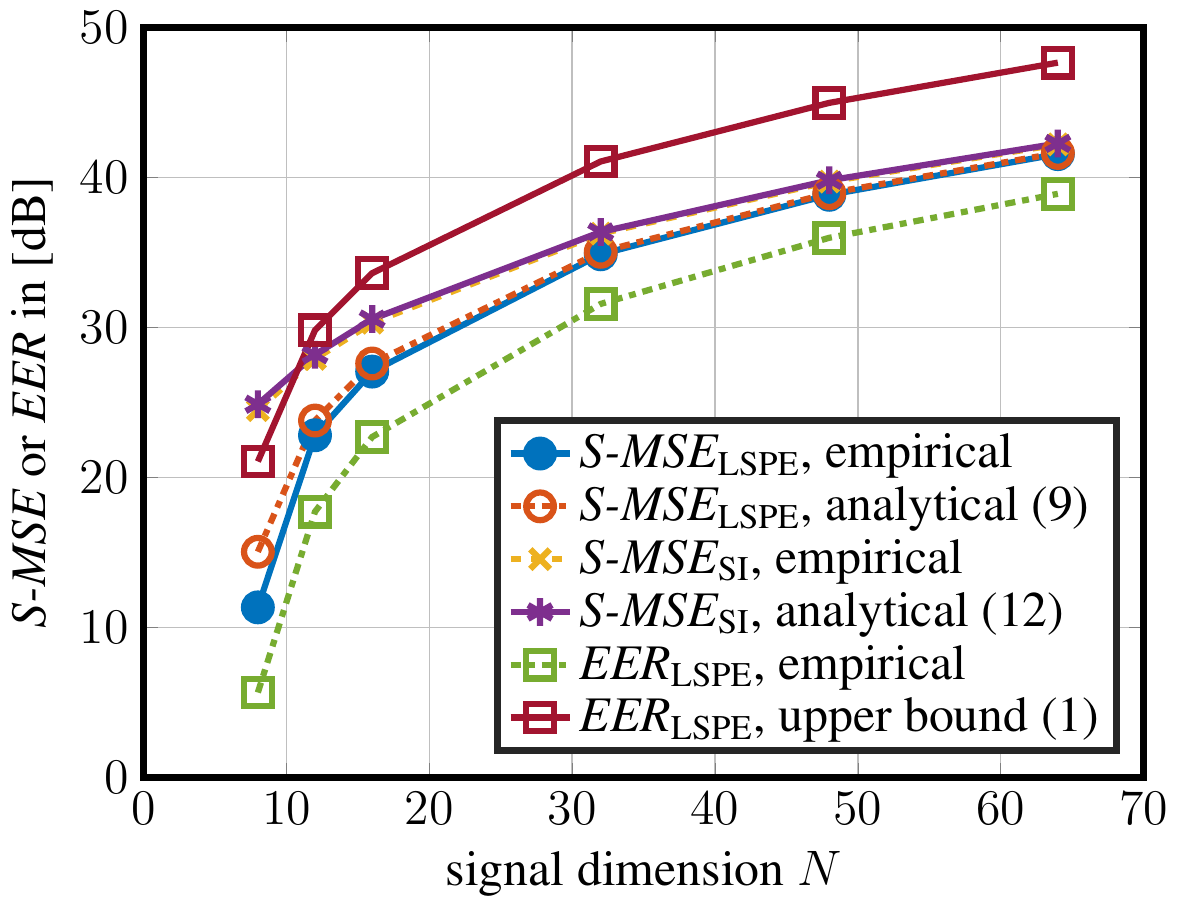}
\caption{Comparison of the analytical and empirical spectral MSE (S-MSE) and estimation error (EER) for LSPEs and spectral initializers (SI) at oversampling ratio $\delta=8$. Our analytical expressions in \fref{thm:MSEofLMSPE} and \fref{prop:SIapproxerror} match the empirical S-MSE; the upper bound in \fref{cor:LMSPEapproxerror} accurately characterizes the empirical EER.}
\vspace{-0.2cm}
\label{fig:comparison}
\end{figure}

\subsection{Real-World Image Recovery}
We finally illustrate the efficacy of LSPEs in a more realistic scenario. In particular, we show results for a real image reconstruction task by using LSPEs and spectral initializers only, i.e., we are \emph{not} using any additional phase retrieval algorithm. 
Our goal is to recover a $16\times16$-pixel and a $40\times40$-pixel image that was captured through a multiple scattering media using the deterministic and highly-structured transmission matrix as detailed in~\cite{metzler2017coherent}.
We compare the proposed LSPEs to the same set of spectral initializers as in \fref{sec:resultsensemble}.
The signal priors are as in \fref{asms:asms2} (LSPE-$\complexset$) and \fref{asms:asms3} (LSPE-Exp).

Figures \ref{fig:TM16} and \ref{fig:TM40} show the recovered images along with the N-MSE values. The proposed LSPEs (often significantly) outperform all spectral initializers in terms of visual quality as well as the N-MSE. This result confirms the observations made in \fref{fig:TM256_MN} that LSPEs outperform existing spectral initializers for structured measurement matrices.
We note that exponential preprocessing for LSPEs does not noticeably improve the N-MSE (over LSPE-$\mathbb{C}$) in this setting since correlations in the transmission measurement matrix are dominating the recovery performance.

\newcommand{\figurescale}{0.6}


\captionsetup[figure]{textfont=it,font=normalsize}
\begin{figure}[t]
\centering
\captionsetup[sub]{textfont=it,font=footnotesize,justification=centering}
\begin{subfigure}[b]{0.32\columnwidth}
\centering 
\includegraphics[width=\figurescale\columnwidth]{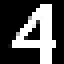}
\caption{original \phantom{asasdasdfasdf}}
\end{subfigure}
\begin{subfigure}[b]{0.32\columnwidth}
\centering 
\includegraphics[width=\figurescale\columnwidth]{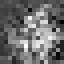}
\caption{amplitude \mbox{$\textit{N-MSE}=0.4927$}}
\end{subfigure} 
\begin{subfigure}[b]{0.32\columnwidth}
\centering 
\includegraphics[width=\figurescale\columnwidth]{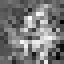}
\caption{optimal \mbox{$\textit{N-MSE}=0.4833$}}
\end{subfigure} \\[0.3cm]
%
%
\begin{subfigure}[b]{0.32\columnwidth}
\centering 
\includegraphics[width=\figurescale\columnwidth]{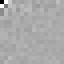}
\caption{orthogonal \mbox{$\textit{N-MSE}=0.6850$}}
\end{subfigure}
\begin{subfigure}[b]{0.32\columnwidth}
\centering 
\includegraphics[width=\figurescale\columnwidth]{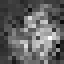}
\caption{spectral \mbox{$\textit{N-MSE}=0.4764 $}}
\end{subfigure} 
\begin{subfigure}[b]{0.32\columnwidth}
\centering 
\includegraphics[width=\figurescale\columnwidth]{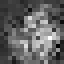}
\caption{truncated \mbox{$\textit{N-MSE}=0.4764$}}
\end{subfigure}  \\[0.3cm]
%
%
\begin{subfigure}[b]{0.32\columnwidth}
\centering 
\includegraphics[width=\figurescale\columnwidth]{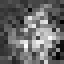}
\caption{weighted  \mbox{$\textit{N-MSE}=0.4797$}}
\end{subfigure}
\begin{subfigure}[b]{0.32\columnwidth}
\centering 
\includegraphics[width=\figurescale\columnwidth]{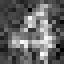}
\caption{LSPE-$\mathbb{C}$ \mbox{$\textit{N-MSE}=0.3377$}}
\end{subfigure} 
\begin{subfigure}[b]{0.32\columnwidth}
\centering 
\includegraphics[width=\figurescale\columnwidth]{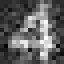}
\caption{LSPE-Exp \mbox{$\textit{N-MSE}=0.2928$}}
\end{subfigure}
\caption{Recovery of a  $16 \times 16$ image from with $ M=5  N$  measurements captured through a scattering medium without the use of a phase retrieval algorithm. LSPEs outperform all initializers for structured measurements.}
\label{fig:TM16}
\end{figure}


\captionsetup[figure]{textfont=it,font=normalsize}
\begin{figure}[t]
\centering
\captionsetup[sub]{textfont=it,font=footnotesize,justification=centering}
\begin{subfigure}[b]{0.32\columnwidth}
\centering 
\includegraphics[width=\figurescale\columnwidth]{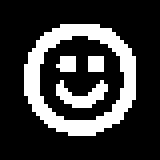}
\caption{original \phantom{asasdasdfasdf}}
\end{subfigure}
\begin{subfigure}[b]{0.32\columnwidth}
\centering 
\includegraphics[width=\figurescale\columnwidth]{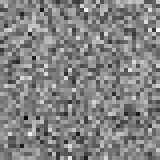}
\caption{amplitude \mbox{$\textit{N-MSE}=0.7010$}}
\end{subfigure} 
\begin{subfigure}[b]{0.32\columnwidth}
\centering 
\includegraphics[width=\figurescale\columnwidth]{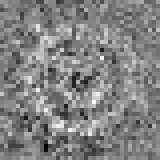}
\caption{optimal \mbox{$\textit{N-MSE}=0.5849 $}}
\end{subfigure} \\[0.3cm]
%
%
\begin{subfigure}[b]{0.32\columnwidth}
\centering 
\includegraphics[width=\figurescale\columnwidth]{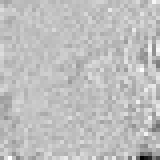}
\caption{orthogonal \mbox{$\textit{N-MSE}=0.7028$}}
\end{subfigure}
\begin{subfigure}[b]{0.32\columnwidth}
\centering 
\includegraphics[width=\figurescale\columnwidth]{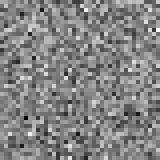}
\caption{spectral \mbox{$\textit{N-MSE}=0.7016$}}
\end{subfigure} 
\begin{subfigure}[b]{0.32\columnwidth}
\centering 
\includegraphics[width=\figurescale\columnwidth]{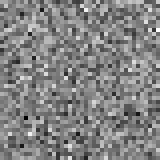}
\caption{truncated \mbox{$\textit{N-MSE}= 0.7020$}}
\end{subfigure}  \\[0.3cm]
%
%
\begin{subfigure}[b]{0.32\columnwidth}
\centering 
\includegraphics[width=\figurescale\columnwidth]{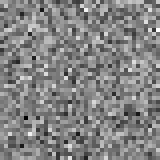}
\caption{weighted  \mbox{$\textit{N-MSE}=0.7013$}}
\end{subfigure}
\begin{subfigure}[b]{0.32\columnwidth}
\centering 
\includegraphics[width=\figurescale\columnwidth]{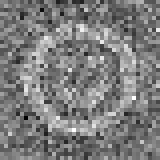}
\caption{LSPE-$\mathbb{C}$ \mbox{$\textit{N-MSE}=0.4920$}}
\end{subfigure} 
\begin{subfigure}[b]{0.32\columnwidth}
\centering 
\includegraphics[width=\figurescale\columnwidth]{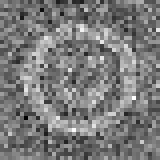}
\caption{LSPE-Exp \mbox{$\textit{N-MSE}=0.4896$}}
\end{subfigure}
\caption{Recovery of a  $40 \times 40$ image from with $ M=10  N$  measurements captured through a scattering medium without the use of a phase retrieval algorithm. LSPEs outperform all initializers for structured measurements.}
\label{fig:TM40}
\end{figure}

%% file: 5-conclusions.tex

\section{Conclusions}
\label{sec:conclusions}

We have proposed a novel class of estimators, called \emph{linear spectral estimators}~(LSPEs), which are suitable for the recovery of signals from general nonlinear measurement systems.
We have developed nonasymptotic and deterministic performance guarantees for LSPEs that provide accurate bounds on the  estimation error, especially for structured or low-dimensional measurement systems.
To demonstrate the efficacy of LSPEs in practice, we have applied them to complex-valued phase retrieval problems, in which LSPEs can be used to compute accurate signal estimates or initialization vectors for other convex or nonconvex phase retrieval algorithms.
We have shown that properly preprocessing the nonlinear measurements can further improve the performance of LSPEs in practical scenarios.
Our simulations with synthetic and real data have shown that LSPEs are able to significantly outperform existing spectral initializers, especially for low-dimensional problems, for structured measurement matrices, or for large oversampling ratios.

There are many avenues for future work.
First, one could derive LSPEs for the asymptotically-optimal preprocessing function in \fref{eq:optimalpreprocessor} or for other commonly used functions, which may lead to further performance improvements. 
Second, the proposed error analysis could be used to generate improved measurement matrices. 
Third, an exploration of LSPEs for other nonlinearities that arise in machine learning and signal processing applications is left for future work.

%% file: 6-appendices-smushed.tex
%


\section{Proof of \fref{thm:LMSPE}}
\label{app:LMSPE}

The proof proceeds in two steps detailed as follows.

\paragraph{Mean Matrix}
We first compute the mean matrix~$\bW_0$. Since \fref{eq:optimalestimator} is a quadratic form, we can take the derivative in~$\widetilde\bW_0^H$ and set it to zero, i.e.,  
\begin{align*}
\frac{d}{d\widetilde\bW_0^H}\Ex{}{\left\| \widetilde\bW_0 + \sum_{m=1}^M \setT(y_m) \widetilde\bW_m-\bmx\bmx^{H}\right\|_F^2} \!= 0.
\end{align*}
Basic matrix calculus yields 
\begin{align} \label{eq:optimalmean}
\widetilde{\bW}_0 =  \textstyle \bK_{\bmx} - \sum_{m=1}^M  \overline\setT(y_m) \widetilde\bW_m
\end{align}
with $\overline\setT(y_m)=\Ex{}{\setT(y_m)}$ and $\bK_{\bmx}=\Ex{}{\bmx\bmx^{H}}$.

\paragraph{Linear Estimation Matrix}
With \fref{eq:optimalmean} and the fact that~\fref{eq:optimalestimator} is a quadratic form in the matrices $\bW_m$, $m=1,\ldots,M$, we take the derivatives in~$\bW_m^H$ and setting them to zero:
\begin{align*}
 {\textstyle  \frac{d}{d\widetilde\bW_m^H}}\mathbb{E} \! \Bigg[\bigg\|\!\!\sum_{m=1}^M\! (\setT(y_m)\!-\!\overline\setT(y_m)) \widetilde\bW_m \!
 \!-\! ( \bmx\bmx^{H}\! \!\!-\! \bK_{\bmx}) \bigg\|_F^2\! \Bigg] \!\!= \!0.
\end{align*}
By interchanging the derivative with expectation and with basic manipulations, we obtain the following set of optimality conditions for $\bW_m$ for $m=1,\dots,M$:
\begin{align} \label{eq:conditions}
& \textstyle \sum_{m'=1}^M \widetilde\bW_{m'} \Ex{}{(\setT(y_m)-\overline\setT(y_m))(\setT(y_{m'})-\overline\setT(y_{m'}))} \notag  \\
& \qquad \quad = \Ex{}{(\setT(y_m)-\overline\setT(y_m))(\bmx\bmx^{H}-\bK_{\bmx})}\!.
\end{align}
In compact matrix form, the above condition reads
\begin{align} \label{eq:matrixcondition}
(\bT\kron\bI_{N\times N}) \underline\bW = \underline\bV,
\end{align}
where we used the following shortcuts: 
\begin{align*}
\bT & = \Ex{}{(\setT(\bmy)-\overline\setT(\bmy))(\setT(\bmy)-\overline\setT(\bmy))^T} \\
\underline\bW & = [ \widetilde\bW_1^T, \ldots, \widetilde\bW_m^T, \ldots, \widetilde\bW_M^T ]^T \\
\bV_m \!& = \Ex{}{(\setT(y_m)\!-\!\overline\setT(y_m))(\bmx\bmx^{H}\!\!-\bK_{\bmx})}\!, m=1,\ldots,M \\
\underline\bV & = [ \bV_1^T, \ldots, \bV_m^T, \ldots, \bV_M^T ]^T.
\end{align*}
The condition in \fref{eq:matrixcondition} can be solved for the estimation matrices in $\underline\bW$ leading to 
$\underline\bW = (\bT^{-1}\kron\bI_{N\times N})  \underline\bV$,
where we require the matrix $\bT$ to be full rank.
To obtain the linear spectral estimator matrix, we simplify as 
\begin{align*}
\bD_\bmy & = \bK_{\bmx} + ((\setT(\bmy)-\overline\setT(\bmy))^T \kron \bI_{N\times N}) \underline\bW \\
& = \textstyle  \bK_{\bmx} + \sum_{m=1}^M t_m \bV_m, \notag
\end{align*}
where we define the vector $ \bmt = \bT^{-1}(\setT(\bmy)-\overline\setT(\bmy))$.

\section{Proof of \fref{thm:MSEofLMSPE}}
\label{app:MSEofLMSPE}
To compute the spectral MSE in \fref{eq:MSE}, we simplify  
\begin{align*}  
\textit{S-MSE} & =\textstyle  \Ex{}{\left\|  \sum_{m=1}^M t_m \bV_m-(\bmx\bmx^{H}-\bK_{\bmx})\right\|_F^2}\!.
\end{align*}
We expand this expression into four terms
\begin{align}
& \textstyle \Ex{}{ \left\|  \sum_{m=1}^M t_m \bV_m-(\bmx\bmx^{H}-\bK_{\bmx})\right\|_F^2} \notag \\
& = \textstyle  \Ex{}{\left\|  \sum_{m=1}^M t_m \bV_m\right\|_F^2} \label{eq:proofterm1} \\
& \textstyle \quad +  \Ex{}{\left\|  \bmx\bmx^{H}-\bK_{\bmx}\right\|_F^2} \notag  \\
& \textstyle \quad - \Ex{}{ \tr\!\left( (\bmx\bmx^{H}-\bK_{\bmx})^H\left(\sum_{m=1}^M t_m \bV_m\right) \right) }  \label{eq:proofterm2} \\
&  \textstyle \quad - \Ex{}{ \tr\!\left( \left(\sum_{m=1}^M t_m \bV_m\right)^H(\bmx\bmx^{H}-\bK_{\bmx})  \right) } \label{eq:proofterm3}
\end{align}
and simplify each expression individually. We start with~\fref{eq:proofterm1} and use the fact that
\begin{align*}
\textstyle \sum_{m=1}^M t_m \bV_m = ((\setT(\bmy)-\overline\setT(\bmy))^T \bT^{-1}\kron\bI_{N\times N})  \underline\bV
\end{align*}
and rewrite the quantity within expectation as follows:
\begin{align*}
&  \underline\bV^H ( \bT^{-1} (\setT(\bmy)-\overline\setT(\bmy)) \kron\bI_{N\times N}) \\
& \quad \times   ((\setT(\bmy)-\overline\setT(\bmy))^T \bT^{-1}\kron\bI_{N\times N})  \underline\bV \\
 & = \underline\bV^H (( \bT^{-1} (\setT(\bmy)-\overline\setT(\bmy)) \\
 & \quad \times  (\setT(\bmy)-\overline\setT(\bmy))^T \bT^{-1}) \kron\bI_{N\times N}) \underline\bV.
\end{align*}
We now evaluate the expectation which leads to
\begin{align*}
\textstyle \Ex{}{\left\|  \sum_{m=1}^M t_m \bV_m\right\|_F^2} = \tr\!\left( \underline\bV^H ( \bT^{-1} \kron\bI_{N\times N}) \underline\bV \right)
\end{align*}
or, equivalently, to 
\begin{align*}  
\textstyle \Ex{}{\left\|  \sum_{m=1}^M \!t_m \! \bV_m\right\|_F^2} \!\!=\displaystyle  \!\!\sum_{m=1}^M\! \sum_{m'=1}^M [ \bT^{-1}]_{m,m'} \tr\!\left( \bV^H_m \bV_{m'} \right)\!.
\end{align*}
We next will simplify \fref{eq:proofterm2}. Recall that
\begin{align*}
\textstyle t_m = \sum_{m'=1}^M [ \bT^{-1}]_{m,m'} (\setT(y_{m'})-\overline\setT(y_{m'})),
\end{align*}
which enables us to write \fref{eq:proofterm2} as
\begin{align}
& \textstyle \Ex{}{ \tr\!\left( (\bmx\bmx^{H}-\bK_{\bmx})^H\left(\sum_{m=1}^M t_m \bV_m\right) \right) } \notag \\ 
%
&\textstyle =  \sum_{m=1}^M \sum_{{m'}=1}^M [ \bT^{-1}]_{m,{m'}}  \notag\\
&\textstyle \quad \times \tr\!\left(  \Ex{}{(\bmx\bmx^{H}-\bK_{\bmx})^H(\setT(y_{m'})-\overline\setT(y_{m'}))} \bV_m \right) \notag  \\
& \textstyle=  \sum_{m=1}^M \sum_{{m'}=1}^M [ \bT^{-1}]_{m,{m'}}    \tr\!\left( \bV_{m'}^H \bV_m \right)\!. \label{eq:proofcombine2}
\end{align}
Seeing as \fref{eq:proofterm3} is the Hermitian conjugate of \fref{eq:proofterm2}, we have
\begin{align}
& \textstyle \Ex{}{ \tr\!\left( \left(\sum_{m=1}^M t_m \bV_m\right)^H(\bmx\bmx^{H}-\bK_{\bmx})  \right) }\notag \\
%
%
& \textstyle = \sum_{m=1}^M \sum_{{m'}=1}^M [ \bT^{-1}]^*_{m,{m'}} \tr\!\left(  \bV_m^H  \bV_{m'} \right)\!.  \label{eq:proofcombine3}
\end{align}

Combining  all these terms yield the spectral MSE 
\begin{align}
  \textit{S-MSE} & =  \textstyle\Ex{}{ \left\|  \sum_{m=1}^M t_m \bV_m-(\bmx\bmx^{H}-\bK_{\bmx})\right\|_F^2} \notag \\
& \textstyle=  \bC_{\bmx\bmx^{H}}-  \tr\!\left( \underline\bV^H ( \bT^{-1} \kron\bI_{N\times N}) \underline\bV \right)\!. \notag
\end{align}
with $\bC_{\bmx\bmx^{H}}=\Ex{}{\left\|  \bmx\bmx^{H}-\bK_{\bmx}\right\|_F^2}$.

\section{Proof of \fref{cor:LMSPEapproxerror}}
\label{app:LMSPEapproxerror}
We bound the estimation error with the spectral MSE of the LSPE as follows. For a given instance, we have
\begin{align}
&  \| \hat\bmx\hat\bmx^H- \bmx\bmx^{H} \|_F^2   =   \| \hat\bmx\hat\bmx^H- \bD_\bmy +\bD_\bmy - \bmx\bmx^{H} \|_F^2 \notag  \\
 &  \quad  \overset{\text{(a)}}{\leq}   2  \| \hat\bmx\hat\bmx^H- \bD_\bmy\|_F^2 + 2\| \bD_\bmy - \bmx\bmx^{H} \|_F^2 \notag \\
 &  \quad \overset{\text{(b)}}{\leq} 4 \| \bD_\bmy - \bmx\bmx^{H} \|_F^2, \notag
\end{align}
where (a) follows from the squared triangle inequality and (b) because $\hat\bmx \hat\bmx^H$ is the best rank-1 approximation of $\bD_\bmy$. Averaging over all instances finally yields
\begin{align*}
 \Ex{}{ \| \hat\bmx\hat\bmx^H- \bmx\bmx^{H} \|_F^2 } \leq 4\,\textit{S-MSE}_\text{LSPE}. 
\end{align*}

%% file: 7-supplementary.tex
\section{Proof of \fref{prop:SIapproxerror}}
\label{app:SIapproxerror}

Our goal is to first evaluate the S-MSE of the unnormalized spectral initializer in \fref{eq:commonspectralinitializer}
\begin{align*}
\textit{US-MSE}_\text{SI} = \Ex{}{ \left\| \beta \sum_{m=1}^M \setT(y_m) \bma_m\bma_m^H - \bmx\bmx^{H}\right\|_F^2 } 
\end{align*}
and then minimize the resulting expression over the parameter $\beta$. The unnormalized spectral MSE 
can be expanded into the following form:
\begin{align*}
& |\beta|^2 \sum_{m=1}^M\sum_{m'=1}^M \Ex{}{ \setT(y_m) \setT(y_{m'}) } \tr(\bma_m\bma_m^H \bma_{m'}\bma_{m'}^H)  \\
& -   \beta^* \sum_{m=1}^M \tr\left( \bma_m\bma_m^H  \Ex{}{  \setT(y_m) \bmx\bmx^{H} } \right) \\
& -    \beta \sum_{m=1}^M \tr\left( \Ex{}{ \bmx\bmx^{H} \setT(y_m)} \bma_m\bma_m^H  \right)   \\
& + \Ex{}{\|\bmx\bmx^{H}\|^2_F}.
\end{align*}
By using the definitions 
\begin{align*}
\widetilde\bV_m & =\Ex{}{\setT(y_m)\bmx\bmx^{H}}, m=1,\ldots,M, \\
\widetilde\bT& =\Ex{}{\setT(\bmy)\setT(\bmy)^T},
\end{align*}
we can simplify the above expression into
 \begin{align}
& |\beta|^2 \sum_{m=1}^M\sum_{m'=1}^M \widetilde T_{m,m'} |\bma_m^H \bma_{m'}|^2  + \Ex{}{\|\bmx\bmx^{H}\|^2_F} \notag \\
&  -   \beta^* \sum_{m=1}^M \!\tr\left( \bma_m\bma_m^H  \widetilde\bV_m \right) \!- \beta\! \sum_{m=1}^M \!\tr\left( \widetilde\bV_m^H \bma_m\bma_m^H  \right)\!.  \label{eq:proofsimplifiedMSESI}
\end{align}
We can now find the optimal parameter for $\beta$ by taking the derivative with respect to $\beta^*$ and setting the expression to zero. The resulting optimal scaling parameter is given by
 \begin{align*}
 \hat\beta   =  \frac{\sum_{m=1}^M \tr\left( \bma_m\bma_m^H  \widetilde\bV_m \right)}{\sum_{m=1}^M\sum_{m'=1}^M \widetilde T_{m,m'} |\bma_m^H \bma_{m'}|^2}.
\end{align*}
We now plug in $\hat\beta$ into the expression \fref{eq:proofsimplifiedMSESI}, which yields
 \begin{align}
& \textit{S-MSE}_\text{SI} =  \left|\frac{\sum_{m=1}^M \tr\!\left( \bma_m\bma_m^H  \widetilde\bV_m \right)}{\sum_{m=1}^M\sum_{m'=1}^M \widetilde T_{m,m'} |\bma_m^H \bma_{m'}|^2} \right|^2 \\
&\quad \times  \sum_{m=1}^M\sum_{m'=1}^M \widetilde T_{m,m'} |\bma_m^H \bma_{m'}|^2  \notag \\
& -   \frac{\sum_{m=1}^M \tr\!\left( \widetilde\bV_m^H\bma_m\bma_m^H   \right)}{\sum_{m=1}^M\sum_{m'=1}^M \widetilde T^*_{m,m'} |\bma_m^H \bma_{m'}|^2} \sum_{m=1}^M \tr\left( \bma_m\bma_m^H  \widetilde\bV_m \right) \notag \\
& -   \frac{\sum_{m=1}^M \tr\!\left( \bma_m\bma_m^H  \widetilde\bV_m \right)}{\sum_{m=1}^M\sum_{m'=1}^M \widetilde T_{m,m'} |\bma_m^H \bma_{m'}|^2}  \sum_{m=1}^M \tr\left( \widetilde\bV_m^H \bma_m\bma_m^H  \right)  \notag \\
& + \Ex{}{\|\bmx\bmx^{H}\|^2_F}.  \notag
\end{align}
This expression can be simplified further to obtain:
 \begin{align*}
\textit{S-MSE}_\text{SI} = & \frac{ \left| \sum_{m=1}^M \tr\left( \bma_m\bma_m^H  \widetilde\bV_m \right) \right|^2 }{\sum_{m=1}^M\sum_{m'=1}^M \widetilde T^*_{m,m'} |\bma_m^H \bma_{m'}|^2}    \notag \\
& -   \frac{ \left| \sum_{m=1}^M \tr\left( \bma_m\bma_m^H  \widetilde\bV_m \right) \right|^2 }{\sum_{m=1}^M\sum_{m'=1}^M \widetilde T^*_{m,m'} |\bma_m^H \bma_{m'}|^2}  \notag \\
& -   \frac{\left| \sum_{m=1}^M \tr\left( \bma_m\bma_m^H  \widetilde\bV_m \right) \right|^2 }{\sum_{m=1}^M\sum_{m'=1}^M \widetilde T_{m,m'} |\bma_m^H \bma_{m'}|^2}   \notag \\
& + \Ex{}{\|\bmx\bmx^{H}\|^2_F},  \notag \\
& = R_{\bmx\bmx^H}\! - \frac{\left| \sum_{m=1}^M \bma_m^H \widetilde{\bV}_m \bma_m \right|^2}{ \sum_{m=1}^M\sum_{m'=1}^M \widetilde{T}_{m,m'}|\bma_m^H\bma_{m'}|^2},
\end{align*}
which is what we wanted to show in \fref{eq:MSESI}.

\section{Real-Valued Phase Retrieval}
\label{app:realvaluedphase}
We now focus on the case where the signal vector $\bmx$ to be recovered and the measurement matrix $\bA$ are both real-valued. We derive the LSPE by using the following assumptions, which are reasonable for phase retrieval problems.

\begin{asms} \label{asms:asms1}
Let $\setH=\reals$. 
Assume square measurements $f(z)=z^2$ and the identity preprocessing function \mbox{$\setT(y)=y$}. 
Assume that the signal vector $\bmx\in\reals^N$ is i.i.d.\ zero-mean Gaussian distributed with covariance matrix $\bC_{\bmx} = \sigma_x^2 \bI_N$, i.e.,  $\bmx \sim \setN(\bZero_{N\times1}, \sigma_x^2 \bI_N)$; the parameter~$\sigma_x^2$ denotes the signal variance.
Assume that the signal noise vector $\bme^z$ is zero-mean Gaussian with covariance matrix $\bC_{\bme^z}$, i.e., $\bme^z \sim \setN(\bZero_{M\times 1},\bC_{\bme^z})$, and the measurement noise vector $\bme^y$ is Gaussian with mean~$\bar\bme^y$ and covariance matrix $\bC_{\bme^y}$, i.e.,  $\bme^y \sim \setN(\bar\bme^y,\bC_{\bme^y})$. Furthermore assume that $\bmx$, $\bme^z$, and $\bme^y$ are independent. 
\end{asms}

Under these assumptions, we can derive the following LSPE which we call LSPE-$\reals$; the detailed derivations of this spectral estimator are given in \fref{app:phaseinitR}.

\begin{esti}[LSPE-$\reals$] \label{esti:phaseinitR}
Let \fref{asms:asms1} hold. Then, the spectral estimation matrix is given by
\begin{align} \label{eq:esti1}
\bD_\bmy^\reals = \bK_\bmx+  \sum_{m=1}^M t_m \bV_m,
\end{align}
where $\bK_\bmx =\sigma_x^2 \bI_N $, 
the vector $\bmt\in\reals^M$ is given by the solution to the linear system $\bT \bmt  = \bmy-\overline\bmy$ with
\begin{align*}
\overline\bmy & =\diag(\bC_{\bmz})+\bar\bme^y \\
\bC_\bmz & = \sigma_x^2 \bA \bA^T+\bC_{\bme^z} \\
\bT &= 2 \bC_{\bmz} \odot \bC_{\bmz}+\bC_{\bme^y}
\end{align*}
and $\bV_m=  2 \sigma_x^4 \bma_m  \bma_m^T $, $m=1,\ldots,M$.
The spectral estimate  $\hat\bmx$ is given by the (scaled) leading eigenvector of $\bD_\bmy^\reals $ in \fref{eq:esti1}. Furthermore, the S-MSE is given by \fref{thm:MSEofLMSPE}.
\end{esti}

\section{Derivation of \fref{esti:phaseinitR}}
\label{app:phaseinitR}

We now use  \fref{thm:LMSPE} to derive \fref{esti:phaseinitR} under \fref{asms:asms1}. To this end, we require the three quantities: $\overline\setT(\bmy)$, $\bT$, and~$\bV_m$, $m=1,\ldots,M$, which we derive separately. 

\paragraph{Computing $\overline\setT(\bmy)$}
To compute the real-valued vector
\begin{align} \label{eq:quantityTbar}
\overline\setT(\bmy)=\Ex{}{\setT(\bmy)}, 
\end{align}
we need the following result on the bivariate folded normal distribution developed in \citep[Sec.~3.1]{foldedNormalMoments}.
\begin{lem}\label{lem:foldedvar}
	Let  $[u_1,u_2] \sim \setN(\bm\mu,\bSigma)$ be a pair of real-valued jointly Gaussian random variables with covariance matrix 
	\begin{align*}	
	\bSigma=\left[\begin{array}{ll}\sigma_1^2 & \sigma_{1,2}^2\\\sigma_{1,2}^2&\sigma_2^2\end{array}\right]\!.
	\end{align*}
	Then, for $m=1,2$, the pair of random variables $(\nu_1, \nu_2)$ with $\nu_1=u_1^2$ and $\nu_2=u_2^2$ follows the bivariate folded normal distribution with the following (centered) moments: 
	\begin{align*}
	\bar{\nu}_m&=\Ex{}{u_m^2}=\sigma_m^2+\mu_m^2\\
	[\bC_{\bm\nu}]_{1,2}&=\Ex{}{(\nu_{1}-\bar{\nu_{1}})(\nu_{2}-\bar{\nu_{2}})} \\ 
	&  =4\mu_{1}\mu_{2}\sigma_{1,2}^2+2\sigma_{1,2}^4\\
	[\bC_{\bm\nu}]_{1,1}&=\Ex{}{(\nu_1-\bar{\nu_1})^2}=2 \sigma_1^4+4\mu_1^2\sigma_1^2.
	\end{align*}
\end{lem}
Let $\bar\bmz=\Ex{}{\bmz}$ denote the mean vector and $\bC_\bmz=\bA \bC_{\bmx} \bA^H+\bC_{\bme^z} = \sigma_x^2 \bA \bA^H+\bC_{\bme^z}$ the covariance matrix of the ``phased'' measurements $\bmz=\bA\bmx+\bme^z$.
Then, by defining $\sigma_m^2 = [\bC_\bmz]_{m,m}$, we can compute the $m$th entry $\overline\setT(y_m)$ using \fref{lem:foldedvar} as follows:
\begin{align}\label{eq:ymeansq}
\overline\setT(y_m)= \bar{y}_m= \Ex{}{|z_m|^2+n^y_m} = \sigma_m^2 + \bar{e}^y_m.
\end{align}
Hence, in compact vector notation we have
\begin{align} \label{eq:ymeansqvector}
\overline\setT(\bmy)=\bar\bmy = \diag(\bC_{\bmz})+\bar\bme^y.
\end{align}

\paragraph{Computing $\bT$}
To compute the real-valued matrix
\begin{align}
\bT & = \Ex{}{(\setT(\bmy)-\overline\setT(\bmy))(\setT(\bmy)-\overline\setT(\bmy))^T} \notag \\
& =  \Ex{}{\setT(\bmy)\setT(\bmy)^T} - \overline\setT(\bmy)\overline\setT(\bmy)^T, \label{eq:quantityTmatrix}
\end{align}
we only need to compute the matrix $\Ex{}{\setT(\bmy)\setT(\bmy)^T}$ as the vector $\overline\setT(\bmy)$ was computed in \fref{eq:ymeansqvector}. 
We compute this matrix entry-wise as  
\begin{align*}
T_{m,m'} & = \Ex{}{(\setT(y_m)-\overline\setT(y_m))(\setT(y_{m'})-\overline\setT(y_{m'}))} \\
&=\Ex{}{y_m y_{m'}^*}-\bar y_m \bar y_{m'}^*\\
&\stackrel{\text{(a)}}{=}\Ex{}{(|z_m|^2+e^y_m)(|z_{m'}|^2+e^y_{m'}) } \\
& \quad -(\sigma_m^2 + \bar{e}^y_m)(\sigma_{m'}^2 + \bar{e}^y_{m'})\\
&=\Ex{}{|z_m|^2 |z_{m'}|^2 }- \sigma_{m'}^2\sigma_m^2 + [\bC_{\bme^y}]_{m,{m'}},
\end{align*}
where (a) follows from \fref{eq:ymeansq}. The only unknown term in the above expression is $ \Ex{}{|z_m|^2|z_{m'}|^2 }$. This term is the second moment of the random vector $[|z_m|^2,|z_{m'}|^2]$, which follows a bivariate folded normal distribution.
For $m \neq m'$, \fref{lem:foldedvar} yields
\begin{align*}
\Ex{}{|z_m|^2 |z_{m'}|^2 }  = \sigma_m^2\sigma_{m'}^2+2 \sigma^4_{m,m'}
\end{align*}
with $\sigma^2_{m,m'}=[\bC_\bmz]_{m,m'}$. 
For $m = m'$, \fref{lem:foldedvar} yields
\begin{align*}
\Ex{}{|y_m|^2} =  \Ex{}{|z_m|^4} = 3 \sigma_m^4  .
\end{align*}
Hence, we have 
\begin{align*}
T_{m,m'} & = [\bC_{\bme^y}]_{m,m'}+\left\{\begin{array}{ll}
2  \sigma_{m,m'}^4 & \text{if } m \neq m' \\
 2 \sigma_m^4  & \text{if } m = m',
\end{array}\right.
\end{align*}
which can be written in compact matrix form as
\begin{align*}
\bT=2 \bC_{\bmz} \odot \bC_{\bmz}+\bC_{\bme^y}.
\end{align*}

\paragraph{Computing $\bV_m$}
To compute the  matrices
\begin{align}
\bV_m & = \Ex{}{(\setT(y_m)-\overline\setT(y_m))(\bmx\bmx^{H}-\bK_{\bmx})} \notag \\
 & = \Ex{}{\setT(y_m)\bmx\bmx^{H}} - \overline{\setT}(y_m) \bK_{\bmx} \label{eq:quantityVmmatrix}
\end{align}
for $m=1,\ldots,M$, we only need to compute the complex-valued matrix $\Ex{}{\setT(y_m)\bmx\bmx^{H}}$ as the two other quantities $\bK_\bmx=\Ex{}{\bmx\bmx^H}$ and  $\overline{\setT}(y_m)$ are known. 
We compute this matrix entry-wise as 
\begin{align*}
[\bV_m]_{n,n'} & =  \Ex{}{(\setT(y_m)-\overline\setT(y_m))x_n x^*_{n'}} \\
& = \Ex{}{y_m x_n x^*_{n'}} - \bar y_m [\bC_{\bmx} ]_{n,n'}.
\end{align*}
Since $\bar y_m$ is known from \fref{eq:ymeansq}, we focus on computing  
\begin{align}\nonumber
& \Ex{}{y_m x_n x^*_{n'}} \\\nonumber
& = \mathbb{E} \Bigg[ \Bigg(\!\!\bigg(\sum_{j=1}^{N}\!\! A_{m,j}^* x_j^*  +e^z_{m}\bigg) \\\nonumber
& \quad \times \bigg(\!\sum_{j'=1}^{N} \!\! A_{m,j'} x_{j'} \! + \! e^z_{m}\bigg) +  e^y_m \! \Bigg)  x_n x^*_{n'}\Bigg]  \\\nonumber
&=\mathbb{E} \Bigg[ \bigg(\sum_{j=1}^{N} A_{m,j}^* x_j^* \sum_{j'=1}^{N} A_{m,j'} x_{j'}\bigg) x_n x^*_{n'} \Bigg] \\\nonumber
&\quad + \Ex{}{|e^z_{m}|^2 x_n x^*_{n'}}+\Ex{}{e^y_{m} x_n x^*_{n'}}\\\label{eq:Velements}
&= \sum_{j=1}^{N} A_{m,j}^*  \sum_{j'=1}^{N} A_{m,j'}  \Ex{}{x_j^* x_{j'} x_n x^*_{n'}}\\\nonumber
&\quad + ([\bC_{\bme^z}]_{m,m} + \bar{e}^y_m)[\bC_{\bmx} ]_{n,n'}.
\end{align}
The only unknown in the above expression is the double summation in \fref{eq:Velements}.
Since we assumed that the entries of the signal vector $\bmx$ are i.i.d., most of the terms in this summation are zero. For $n \neq n'$, there are only two nonzero terms, corresponding to the cases of $(j,j')=(n,n')$ and $(j,j')=(n',n)$.  Thus, for $n \neq n'$ we have
\begin{align}\label{eq:V_nondiag}\nonumber
& \sum_{j=1}^{N} A_{m,j}^*  \sum_{j'=1}^{N} A_{m,j'}  \Ex{}{x_j^* x_{j'} x_n x^*_{n'}} \\ 
& \qquad = 2 A_{m,n}^* A_{m,n'} \Ex{}{ |x_n|^2 |x_{n'}|^2} \nonumber \\
& \qquad \stackrel{\text{(b)}}{=} 2 A_{m,n}^* A_{m,n'} [\bC_{\bmx}]_{n,n} [\bC_{\bmx}]_{n',n'},
\end{align}
where (b) follows from \fref{lem:foldedvar}.
For $n = n'$, we have
\begin{align*}
&\sum_{j=1}^{N} A_{m,j}^*  \sum_{j'=1}^{N} A_{m,j'}  \Ex{}{x_j^* x_{j'} x_n x^*_{n}} \\
& =|A_{m,n}|^2 \Ex{}{ |x_n|^4}+\sum_{j\neq n, j=1}^{N} |A_{m,j}|^2  \Ex{}{ |x_j|^2 |x_n|^2}\\
& \stackrel{\text{(c)}}{=} 3 |A_{m,n}|^2 [\bC_{\bmx}]_{n,n}^2+\sum_{j\neq n, j=1}^{N} |A_{m,j}|^2 [\bC_{\bmx}]_{j,j} [\bC_{\bmx}]_{n,n} \\
& = 2 |A_{m,n}|^2 [\bC_{\bmx}]_{n,n}^2+\sum_{j=1}^{N} |A_{m,j}|^2 [\bC_{\bmx}]_{j,j} [\bC_{\bmx}]_{n,n}.
\end{align*}
As for \fref{eq:V_nondiag}, (c) follows from \fref{lem:foldedvar}.
By combining the above results, we have  
\begin{align*}
\bV_m =& \, 2\bC_{\bmx}^H \bma_m  \bma_m^H \bC_{\bmx} + (\bma_m^H \bC_{\bmx}  \bma_m) (\bC_{\bmx}^H \odot \bI) \\
& +([\bC_{\bme^z}]_{m,m}-\sigma_m^2) \bC_{\bmx} = 2 \sigma_x^4 \bma_m  \bma_m^H,
\end{align*}
where $\bma_m^H$ denotes the $m$th row of the matrix $\bA$.

\section{Derivation of \fref{esti:phaseinitC}}
\label{app:phaseinitC}
We now use  \fref{thm:LMSPE} to derive \fref{esti:phaseinitC} under \fref{asms:asms2}. To this end, we require the three quantities: $\overline\setT(\bmy)$, $\bT$, and~$\bV_m$, $m=1,\ldots,M$, which we derive separately. 

\paragraph{Computing $\overline\setT(\bmy)$}
To compute the real-valued vector $\overline\setT(\bmy)=\bar\bmy$  in~\fref{eq:quantityTbar}, we need the following definitions. Let $\bar\bmz=\Ex{}{\bmz}$ denote the mean vector and $\bC_\bmz=\bA \bC_{\bmx} \bA^H+\bC_{\bme^z} = \sigma_x^2 \bA \bA^H+\bC_{\bme^z}$ the covariance matrix of the ``phased'' measurements $\bmz=\bA\bmx+\bme^z$.
Then, using \fref{lem:foldedvar} with the definitions $\bar{\bmz}$ and $\bC_\bmz$, we have
\begin{align}\nonumber
\bar{y}_m&=\Ex{}{|z_m|^2+\bar{e}^{y}_m}= \Ex{}{|z_{m,\mathcal{R}}|^2+|z_{m,\mathcal{I}}|^2+ \bar{e}^{y}_m}\\\label{eq:ymeansq-complex}
&=\sigma^2_{m}+\bar{e}^y_m,
\end{align}
where we have used the definition $\sigma^2_{m}=[\bC_{\bmz}]_{m,m}$. 
Hence, in compact vector notation we have
\begin{align*}
\overline\setT(\bmy)= \bar\bmy = \diag(\bC_\bmz)+\bar{\bme}^y.
\end{align*}

\paragraph{Computing $\bT$}
To compute the real-valued matrix~$\bT$ in~\fref{eq:quantityTmatrix}, we will frequently use the following result. 
Since the vector $\bmz$ is a complex circularly-symmetric jointly Gaussian vector, we can extract the covariance matrices of the real and imaginary parts separately as:
\begin{align}\label{eq:CZ_R}
\Ex{}{\bmz_{\mathcal{I}} \bmz_{\mathcal{I}}^H}&\stackrel{\text{(a)}}{=}\Ex{}{\bmz_{\mathcal{R}} \bmz_{\mathcal{R}}^H}=\frac{1}{2}\Re\{\bC_\bmz\}=\frac{1}{2}\bC_{\bmz,\mathcal{R}}\\\label{eq:CZ_m}
\Ex{}{\bmz_{\mathcal{R}} \bmz_{\mathcal{I}}^H}&=-\Ex{}{\bmz_{\mathcal{I}} \bmz_{\mathcal{R}}^H}=\frac{1}{2}\Im\{\bC_\bmz\}=\frac{1}{2}\bC_{\bmz,\mathcal{I}},
\end{align}
where (a) follows from circular symmetry of the random vector $\bmx$.
We are now ready to compute the individual entries of~$\Ex{}{\setT(\bmy)\setT(\bmy)^T}$ as
\begin{align*}
	T_{m,m'} & = \Ex{}{(\setT(y_m)-\overline\setT(y_m))(\setT(y_{m'})-\overline\setT(y_{m'}))} \\
	&= \Ex{}{(y_m-\bar y_m)(y_{m'}-\bar y_{m'})^*} \\
	&= \Ex{}{y_my_{m'}^*} - \bar y_m \bar y_{m'}^*.
\end{align*}
The quantity $ \bar y_m$ is given by \fref{eq:ymeansq-complex}. Hence, we now compute
\begin{align*}
	& \Ex{}{y_my_{m'}^*} \\
	& = \Ex{}{(|z_m|^2+e^y_m) (|z_m'|^2+e^y_{m'})^* } \\
	&=\Ex{}{\left(|z_{m,\mathcal{R}}|^2+|z_{m,\mathcal{I}}|^2\right)\left(|z_{m',\mathcal{R}}|^2+|z_{m',\mathcal{I}}|^2\right)}\\
	& \quad +[\bC_{\bme^y}]_{m,m}\\
	&=2\Ex{}{|z_{m,\mathcal{R}}|^2|z_{m',\mathcal{R}}|^2}+2\Ex{}{|z_{m,\mathcal{R}}|^2|z_{m',\mathcal{I}}|^2} \\
	& \quad +[\bC_{\bme^y}]_{m,m}.
\end{align*}
The first two terms above are a second moment of the variables $[|z_{m,\mathcal{R}}|^2,|z_{m',\mathcal{R}}|^2]$ and $[|z_{m,\mathcal{R}}|^2,|z_{m',\mathcal{I}}|^2]$, which follow a bivariate folded normal distributions.  We first focus on $[|z_{m,\mathcal{R}}|^2,|z_{m',\mathcal{R}}|^2]$.  With \fref{lem:foldedvar}, we can calculate the moments using the covariance $\Ex{}{\bmz_{\mathcal{R}} \bmz_{\mathcal{R}}^H}$ given in \fref{eq:CZ_R}.  To this end, define $\sigma^2_{m,m',\mathcal{R}}=[\bC_{\bmz,\mathcal{R}}]_{m,m'}$ and $\sigma^2_{m,\mathcal{R}}=[\bC_{\bmz,\mathcal{R}}]_{m,m}$. Thus, we have
\begin{align*}
\Ex{}{|z_{m,\mathcal{R}}|^2|z_{m',\mathcal{R}}|^2} \!& = \!\left\{\begin{array}{ll}
 \!\! \frac{\sigma^2_{m,\mathcal{R}}}{2} \frac{\sigma^2_{m',\mathcal{R}}}{2} + \frac{\sigma^4_{m,m',\mathcal{R}}}{2},  &\!\!\!\! m\neq m' \\
\!\! 3\frac{\sigma^4_{m,\mathcal{R}}}{4},  &\!\!\!\! m=m'.
\end{array}\right.
\end{align*}
Analogously, we can compute $\Ex{}{\bmz_{\mathcal{R}} \bmz_{\mathcal{I}}^H}$ in \fref{eq:CZ_m} from the covariance matrix of $[|z_{m,\mathcal{R}}|^2,|z_{m',\mathcal{I}}|^2]$, with $\sigma^2_{m,m',\mathcal{I}}=[\bC_{\bmz,\mathcal{I}}]_{m,m'}$ and noting that $\sigma^2_{m,\mathcal{I}}=[\bC_{\bmz,\mathcal{I}}]_{m,m}=0$ as
\begin{align*}
\Ex{}{|z_{m,\mathcal{R}}|^2|z_{m',\mathcal{I}}|^2} \!& =\! \left\{\begin{array}{ll}
 \!\!  \frac{\sigma^2_{m,\mathcal{R}}}{2} \frac{\sigma^2_{m',\mathcal{R}}}{2} + 2 \frac{\sigma^4_{m,m',\mathcal{I}}}{4},  &\!\!\!\! m \neq m' \\
 \!\,3\frac{\sigma^4_{m,\mathcal{R}}}{4},  &\!\!\!\!  m = m'.
\end{array}\right.
\end{align*}
By combining the above results, we have
\begin{align*}
T_{m,m'} \!& =\! \left\{\begin{array}{ll}
\!\!\sigma^2_{m,\mathcal{R}} \sigma^2_{m',\mathcal{R}} + \sigma^4_{m,m',\mathcal{R}}+\sigma^4_{m,m',\mathcal{I}}, & \!\!   m \neq m' \\
\!\!2 \, \sigma^4_{m,\mathcal{R}}, & \!\! m = m',
\end{array}\right.\\
&\quad +[\bC_{\bme^y}]_{m,m'}-  \bar y_m \bar y_{m'}^*\\
 & = [\bC_{\bme^y}]_{m,m'} + \left\{\begin{array}{ll}
 \sigma^4_{m,m',\mathcal{R}}+\sigma^4_{m,m',\mathcal{I}}, & \!m \neq m' \\
\sigma^4_{m,\mathcal{R}}, & \!  m = m',
\end{array}\right.
\end{align*}
which can be written in matrix form as
\begin{align*}
\bT=\bC_\bmz \odot \bC_\bmz^*+\bC_{\bme^y}.
\end{align*}

\paragraph{Computing $\bV_m$}
To compute the matrices~$\bV_m$, $m=1,\ldots,M$, in \fref{eq:quantityVmmatrix}, we need the complex-valued matrix $\Ex{}{\setT(y_m)\bmx\bmx^{H}}$. 
We compute this matrix entry-wise as 
\begin{align*}
	[\bV_m]_{n,n'} & =  \Ex{}{(\setT(y_m)-\overline\setT(y_m))x_n x^*_{n'}} \\
	& = \Ex{}{y_m x_n x^*_{n'}} - \bar y_m [\bC_{\bmx}]_{n,n'}.
\end{align*}
Since $\bar y_m$ is given by \fref{eq:ymeansq-complex}, we only need to compute 
\begin{align}\label{eq:Velements-complex}\nonumber
	& \Ex{}{y_m x_n x^*_{n'}} \\ \nonumber
	& =  \mathbb{E} \Bigg[  \bigg( \Big(\sum_{j=1}^{N} \!\! A_{m,j}^* x_j^*+ e^{z*}_{m} \Big) \\\nonumber
	& \quad \qquad \times  \Big(\sum_{j'=1}^{N}  A_{m,j'} x_{j'}+e^z_{m} \Big)+e^y_m \bigg) x_n x^*_{n'}  \Bigg] \\\nonumber
	&=\sum_{j=1}^{N} A_{m,j}^*  \sum_{j'=1}^{N} A_{m,j'}  \Ex{}{x_j^* x_{j'} x_n x^*_{n'}}\\ \nonumber
	& \quad + \Ex{}{|e^z_{m}|^2 x_n x^*_{n'}}+\Ex{}{e^y_{m} x_n x^*_{n'}}\\ 
		&=\sum_{j=1}^{N} A_{m,j}^*  \sum_{j'=1}^{N} A_{m,j'}  \Ex{}{x_j^* x_{j'} x_n x^*_{n'}}\\ \nonumber
	& \quad + ([\bC_{\bme^z}]_{m,m}+\bar e^y_m)[\bC_{\bmx}]_{n,n'}.
\end{align}
We will first simplify the term 
\begin{align*}
\sum_{j=1}^{N} A_{m,j}^*  \sum_{j'=1}^{N} A_{m,j'}  \Ex{}{x_j^* x_{j'} x_n x^*_{n'}}\!.
\end{align*} 
Since we assumed that the signal vector $\bmx$ has i.i.d.\ zero-mean entries, most of the terms in this summation are zero. For $n \neq n'$, there is only one non-zero term for $(j,j')=(n,n')$. Thus, for $n \neq n'$ we have
\begin{align*}
	& \sum_{j=1}^{N} A_{m,j}^*  \sum_{j'=1}^{N} A_{m,j'}  \Ex{}{x_j^* x_{j'} x_n x^*_{n'}} \\
	& \qquad  \qquad =  A_{m,n}^* A_{m,n'} [\bC_{\bmx}]_{n,n} [\bC_{\bmx}]_{n',n'},
\end{align*}
since the term that corresponds to $(j,j')=(n',n)$, i.e. $A_{m,n'}^* A_{m,n} \Ex{}{x_{n'}^* x^*_{n'}} \Ex{}{x_{n} x_n}$, is zero.

Next, for $n = n'$, we have
\begin{align*}
	&\sum_{j=1}^{N} A_{m,j}^*  \sum_{j'=1}^{N} A_{m,j'}  \Ex{}{x_j^* x_{j'} x_n x^*_{n'}}\\
	& =|A_{m,n}|^2 \Ex{}{ |x_n|^4}+\!\!\sum_{j\neq k, j=1}^{N} \!\! |A_{m,j}|^2  \Ex{}{ |x_j|^2 |x_n|^2}\\
	&= |A_{m,n}|^2 \Ex{}{ |x_{n,\mathcal{R}}|^4}+|A_{m,n}|^2 \Ex{}{|x_{n,\mathcal{I}}|^4}\\
	&\quad + 2|A_{m,n}|^2 \Ex{}{ |x_{n,\mathcal{R}}|^2 |x_{n,\mathcal{I}}|^2} \\
	&\quad +\sum_{j\neq n, j=1}^{N} |A_{m,j}|^2   \\
	&\qquad \times\Ex{}{(|x_{j,\mathcal{R}}|^2+|x_{j,\mathcal{I}}|^2)(|x_{n,\mathcal{R}}|^2+|x_{n,\mathcal{I}}|^2)}\\
	&\stackrel{\text{(a)}}{=} 2|A_{m,n}|^2 \Ex{}{ |x_{n,\mathcal{R}}|^4} \\
	& \quad +  2\sum_{ j=1}^{N} |A_{m,j}|^2  \Ex{}{|x_{j,\mathcal{R}}|^2|x_{n,\mathcal{I}}|^2}\\
	&\quad +2\sum_{j\neq n, j=1}^{N} |A_{m,j}|^2  \Ex{}{|x_{j,\mathcal{R}}|^2|x_{n,\mathcal{R}}|^2}\\
	&\stackrel{\text{(b)}}{=} |A_{m,n}|^2 [\bC_{\bmx}]_{n,n}^2 +\sum_{ j=1}^{N} |A_{m,j}|^2 [\bC_{\bmx}]_{j,j} [\bC_{\bmx}]_{n,n},
\end{align*}
where (a) follows from circular symmetry of $\bmx$ and (b) from \fref{lem:foldedvar}.
By combining the above results, we have  
\begin{align*}
\bV_m & =\bC_{\bmx}^H \bma_m  \bma_m^H \bC_{\bmx} + (\bma_m^H \bC_{\bmx}  \bma_m) (\bC_{\bmx}^H \odot \bI) \\
& \quad +([\bC_{\bme^z}]_{m,m}-\sigma_m^2) \bC_{\bmx} = \sigma_x^4 \bma_m  \bma_m^H.
\end{align*}

\section{Derivation of \fref{esti:phaseinitExp}}
\label{app:phaseinitExp}

We now use  \fref{thm:LMSPE} to derive \fref{esti:phaseinitExp} under \fref{asms:asms3}. To this end, we require the three quantities: $\overline\setT(\bmy)$, $\bT$, and~$\bV_m$, $m=1,\ldots,M$, which we derive separately. 

\paragraph{Computing $\overline\setT(\bmy)$}
To derive an expression for~$\overline\setT(\bmy)$ in~\fref{eq:quantityTbar}, we need the following two results.
\begin{lem} \label{lem:exppreproc}
Let  $\bmu \sim \setC\setN(\bZero_{M\times1},\bSigma)$ be a complex-valued circularly-symmetric  jointly Gaussian random vector with positive definite covariance matrix $\bSigma\in\complexset^{M\times M}$. Then, for the random variable $\nu=\exp(-\bmu^H\bG\bmu )$ with positive definite $\bG\in\complexset^{M\times M}$ and $\bG+\bSigma^{-1}$ positive definite, we have the following result:
\begin{align*}
\Ex{}{\nu} = \frac{1}{|\bG\bSigma+\bI_{M}|}.
\end{align*}
\end{lem}

\begin{proof}
We first expand the expected value into
\begin{align*}
& \Ex{}{\nu} = \Ex{}{\exp(-\bmu^H\bG\bmu)} = \\ 
& \quad \int_{\complexset^M} \exp(-\bmu^H\bG\bmu) \frac{1}{\pi^M|\bSigma|} \exp(-\bmu^H\bSigma^{-1}\bmu) \text{d}\bmu,
\end{align*}
where $|\bSigma|>0$ is the determinant of~$\bSigma$. We can now simplify the above expression as follows:
\begin{align*}
& \int_{\complexset^M} \exp(-\bmu^H\bG\bmu) \frac{1}{\pi^M|\bSigma|} \exp(-\bmu^H\bSigma^{-1}\bmu) \text{d}\bmu \\
& = \int_{\complexset^M}  \frac{1}{\pi^M|\bSigma|} \exp\big(-\bmu^H(\bG+\bSigma^{-1})\bmu\big)   \text{d}\bmu \\
& =\frac{\pi^M|(\bG+\bSigma^{-1})^{-1}|}{\pi^M|\bSigma|} \frac{1}{\pi^M|(\bG+\bSigma^{-1})^{-1}|}  \\
& \qquad \times \int_{\complexset^M}   \exp\big(-\bmu^H(\bG+\bSigma^{-1})\bmu\big)   \text{d}\bmu \\
& = \frac{|(\bG+\bSigma^{-1})^{-1}|}{|\bSigma|} =  \frac{1}{|\bG+\bSigma^{-1}||\bSigma|} = \frac{1}{|\bG\bSigma+\bI|},
\end{align*}
where we also required that $\bG+\bSigma^{-1}$ is positive definite.
\end{proof}

\begin{lem} \label{lem:negmgfofgaussian}
Let $\bmu \sim \setN(\bar\bmu,\bSigma)$ be a real-valued Gaussian random vector with mean $\bar\bmu$ and covariance $\bSigma$, and $\boldsymbol \gamma\in\reals^N$ be a given vector. Then, we have 
\begin{align*}
\Ex{}{\exp(-\boldsymbol\gamma^T\bmu)} = \exp\!\left(-\boldsymbol\gamma^T\bar\bmu+\textstyle\frac{1}{2}\boldsymbol\gamma^T\bSigma\boldsymbol\gamma\right)\!.
\end{align*}
\end{lem}

\begin{proof}
The proof is an immediate consequence of the moment generating function of a Gaussian random vector.
\end{proof}

By considering \fref{lem:exppreproc} and \fref{lem:negmgfofgaussian} for scalar random variables, the $m$th entry of the preprocessed phaseless measurement is given by
\begin{align*}
\overline\setT(y_m) & =\Ex{}{\setT(y_m)} = \Ex{}{\exp(-\gamma |z_m|^2-\gamma [\bme^y]_m)} \\
& = \frac{1}{ \gamma[\bC_\bmz]_{m,m}+1 } \exp \!\left(-\gamma [\bar{\bme}^y]_m + \textstyle \frac{1}{2}\gamma^2[\bC_{\bme^y}]_{m,m}\right)\!.
\end{align*}
We define the following auxiliary vectors
\begin{align}
\bmq_\gamma& =\gamma\diag(\bC_\bmz)+\bOne_{M\times1} \label{eq:qquantity} \\ 
\bmp_\gamma&=\exp\!\left(-\gamma \bar\bme^y + \textstyle\frac{1}{2}\gamma^2\diag(\bC_{\bme^y})\right)\!, \label{eq:pquantity}  
\end{align} 
which enable us to rewrite the above expression in compact vector form as
\begin{align*}
\overline\setT(\bmy) = &\,\,  \bmp_\gamma \oslash \bmq_\gamma.
\end{align*}


\paragraph{Computing $\bT$}
To compute the matrix~$\bT$ in~\fref{eq:quantityTmatrix}, we only need to compute  $\Ex{}{\setT(\bmy)\setT(\bmy)^T}$, which we will compute entry-wise and in two separate steps. 
Concretely, we have
\begin{align*}
\Ex{}{\setT(y_m)\setT(y_{m'})}  = &\, \Ex{}{\exp(-\gamma(|z_m|^2+|z_{m'}|^2))} \\
& \times \Ex{}{\exp(-\gamma([\bme^y]_m+[\bme^y]_{m'}))}\!,
\end{align*}
where we compute both expected values separately.
In the first step, we compute 
\begin{align*}
 \Ex{}{\exp(-\gamma(|z_m|^2+|z_{m'}|^2))}  = \Ex{}{\exp(-\bmu^H \bG \bmu ))}\!,
\end{align*}
with $\bmu=[z_m,z_{m'}]^T$ and $\bG=\bI_2\gamma$. By invoking \fref{lem:exppreproc} with
$[\bSigma]_{m,m'}  = [\bC_\bmz]_{m,m'}$, we obtain
\begin{align*}
&  \Ex{}{\exp(-\gamma(|z_m|^2+|z_{m'}|^2))}  = \frac{1}{|\gamma\bSigma+\bI_2|} \\
& =    \frac{1}{(\gamma[\bC_\bmz]_{m,m}+1) (\gamma[\bC_\bmz]_{m',m'}+1) - \gamma^2 |[\bC_\bmz]_{m,m'}|^2}.
\end{align*}
With the definition of $\bmq_\gamma$ in \fref{eq:qquantity}, we can rewrite the above expression in vector form as
\begin{align*}
& \Ex{}{\exp(-\gamma|\bmz|^2)\exp(-\gamma|\bmz|^2)^T}  \\
& \qquad  \qquad \qquad = \bOne_{M\times M} \oslash ( \bmq_\gamma\bmq^T_\gamma  - \gamma^2 \bC_\bmz\odot \bC_\bmz^*).
\end{align*}
In the second step, we compute
\begin{align*}
\Ex{}{\exp(-\gamma([\bme^y]_m+[\bme^y]_{m'}))} = \Ex{}{\exp(-\boldsymbol\gamma^T\bmu)} 
\end{align*}
with $\bmu=[[\bme^y]_m,[\bme^y]_{m'}]^T$ and $\boldsymbol\gamma^T=[\gamma,\gamma]$. By invoking \fref{lem:negmgfofgaussian} with mean $\bar\bmu=[\bar[\bme^y]_m,[\bar\bme^y]_{m'}]$ and covariance $\bSigma$ given by the entries of the covariance matrix $\bC_{\bme^y}$ associated to the indices $m$ and $m'$, we obtain
\begin{align*}
&\Ex{}{\exp(-\gamma([\bme^y]_m+[\bme^y]_{m'}))} = \exp(-\gamma([\bar\bme^y]_m+[\bar\bme^y]_{m'})) \\
&  \qquad \times  \exp(\textstyle\frac{1}{2}\gamma^2([\bC_{\bme^y}]_{m,m}+[\bC_{\bme^y}]_{m',m'}+2[\bC_{\bme^y}]_{m,m'}) ).
\end{align*}
With the definition of $\bmp_\gamma$ in \fref{eq:pquantity}, we can rewrite the above expression in vector form as
\begin{align*}
&\Ex{}{\exp(-\gamma\bme^y)\exp(-\gamma\bme^y)^T} = (\bmp_\gamma\bmp_\gamma^T) \odot \exp(\gamma^2\textstyle\bC_{\bme^y})
\end{align*}
We furthermore have 
\begin{align*}
\overline\setT(\bmy)\overline\setT(\bmy)^T = ( \bmp_\gamma\bmp_\gamma^T) \oslash   (\bmq_\gamma\bmq_\gamma^T).
\end{align*}
By combining the two steps with the above results, we have
\begin{align*}
 \bT  = &\,  (\bmp_\gamma\bmp_\gamma^T) \odot \big(\!\exp(\gamma^2\bC_{\bme^y}) \oslash ( \bmq_\gamma\bmq^T_\gamma\!  - \gamma^2 \bC_\bmz\odot \bC_\bmz^*) \\  
& - \bOne_{M\times M} \oslash   (\bmq_\gamma\bmq_\gamma^T)\big).
\end{align*}
%


\paragraph{Computing $\bV_m$}
To compute the matrices $\bV_m$, $m=1,\ldots,M$, in \fref{eq:quantityVmmatrix}, 
we only need $\Ex{}{\setT(y_m)\bmx\bmx^{H}}$ which we will compute entry-wise and in two steps.
We have
\begin{align*}
\Ex{}{\setT(y_m)x_nx_{n'}^*} = &\, \Ex{}{\exp(-\gamma|\bma_m^H\bmx+[\bme^z]_m|^2)x_nx_{n'}^*} \\
& \times \Ex{}{\exp(-\gamma[\bme^y]_{m})},
\end{align*}
where we next compute both expected values separately. 
As a first step, we use direct integration to compute the following expected value:
\begin{align*}
& \Ex{}{\exp(-\gamma|\bma_m^H\bmx+[\bme^z]_m|^2)x_n x_{n'}^*}=\! \int_{\complexset^{N+1}}\!\!\! \exp(-\gamma|\bma_m^H\bmx+[\bme^z]_m|^2)  \\
& \quad \times \frac{1}{(\pi\sigma_x^2)^N}\exp\left(-\frac{\|\bmx\|^2}{\sigma_x^2}\right) \\
& \quad \times \frac{1}{\pi\sigma_n^2}\exp\left(-\frac{|[\bme^z]_m|^2}{\sigma_n^2}\right) x_nx_{n'}^* \text{d}\bmx \text{d} [\bme^z]_m.
\end{align*}
We define the following auxiliary quantities:
\begin{align*}
\tilde\bma_m^H & = [\, \bma^H_m, 1\,] \\
\tilde\bmx^T & = [\,\bmx^T,[\bme^z]_m\,] \\
\bC_{\tilde\bmx} & = \left[\begin{array}{cc}
\sigma_x^2 \bI_{N} & \bZero_{N\times 1} \\
\bZero_{1\times N} & \sigma_m^2 
\end{array}\right] \\
\widetilde\bK^{-1} & =\gamma\tilde\bma_m\tilde\bma^H_m+\bC^{-1}_{\tilde\bmx},
\end{align*}
where $\sigma_m^2 = \Ex{}{|[\bme^z]_m|^2} = [\bC_{\bmn^z}]_{m,m}$. 
We now derive the above expectation in compact form as
\begin{align*}
& \Ex{}{\exp(-\gamma|\tilde\bma_m^H\tilde\bmx|^2)\tilde{x}_n \tilde{x}_{n'}^*}=  \\
& =\frac{1}{(\pi\sigma_x^2)^N} \frac{1}{\pi\sigma_n^2} \int_{\complexset^{N+1}} \!\!\!\!\exp(-\gamma|\tilde\bma^H\tilde\bmx|^2\!-\tilde\bmx^H\bC^{-1}_{\tilde\bmx}\tilde\bmx)\tilde{x}_n\tilde{x}_{n'}^* \text{d}\tilde\bmx \\
& = \frac{1}{|\pi \bC_{\tilde\bmx}|}  \int_{\complexset^{N+1}} \!\!\!\exp(-\tilde\bmx^H(\gamma\tilde\bma_m\tilde\bma^H_m+\bC^{-1}_{\tilde\bmx})\tilde\bmx)\tilde{x}_n\tilde{x}_{n'}^* \text{d}\tilde\bmx \\
& = \frac{1}{|\pi\bC_{\tilde\bmx}|}  \int_{\complexset^{N+1}}\!\!\! \exp(-\tilde\bmx^H \widetilde\bK^{-1} \tilde\bmx)\tilde{x}_n\tilde{x}_{n'}^* \text{d}\tilde\bmx,
\end{align*}
where $n=1,\ldots,N+1$, $n'=1,\ldots,N+1$. 
We can further rewrite this expression as 
\begin{align*}
&  \frac{1}{|\pi\bC_{\tilde\bmx}|}  \int_{\complexset^{N+1}} \exp(-\tilde\bmx^H \widetilde\bK^{-1} \tilde\bmx)\tilde{x}_n\tilde{x}_{n'}^* \text{d}\tilde\bmx \\
& \qquad = \frac{|\pi\widetilde\bK|}{|\pi\widetilde\bK||\pi\bC_{\tilde\bmx}|}  \int_{\complexset^{N+1}} \exp(-\tilde\bmx^H \widetilde\bK^{-1} \tilde\bmx)\tilde{x}_n\tilde{x}_{n'}^* \text{d}\tilde\bmx.
\end{align*}
It is now key to realize that 
\begin{align*}
& \frac{1}{|\pi\widetilde\bK|}  \int_{\complexset^{N+1}} \exp(-\tilde\bmx^H \widetilde\bK^{-1} \tilde\bmx)\tilde{x}_n\tilde{x}_{n'}^* \text{d}\tilde\bmx \\
& \qquad \qquad = \Ex{}{\tilde{x}_n\tilde{x}_{n'}^*} = [\widetilde\bK]_{n,n'}
\end{align*}
and hence we have
\begin{align*}
& \Ex{}{\exp(-\gamma|\tilde\bma_m^H\tilde\bmx|^2)\tilde{x}_n \tilde{x}_{n'}^*} \\
& \qquad = \frac{|\widetilde\bK|}{|\bC_{\tilde\bmx}|} [\widetilde\bK]_{n,n'} = \frac{1}{|\widetilde\bK^{-1}||\bC_{\tilde\bmx}|} [\widetilde\bK ]_{n,n'} \\
 & \qquad= \frac{1}{|\gamma\tilde\bma_m\tilde\bma^H_m+\bC_{\tilde\bmx}^{-1}||\bC_{\tilde\bmx}|} [\widetilde\bK ]_{n,n'} \\
 & \qquad=  \frac{1}{|\gamma\tilde\bma_m\tilde\bma^H_m\bC_{\tilde\bmx}+\bI_{N+1}|} [\widetilde\bK ]_{n,n'}.
\end{align*}
We can now use the matrix-determinant lemma to simplify
\begin{align*}
|\gamma\tilde\bma_m\tilde\bma^H_m\bC_{\tilde\bmx}+\bI_{N+1}| & = \gamma\tilde\bma_m^H \bC_{\tilde\bmx}\tilde\bma_m+1 \\
& = \gamma(\sigma_x^2\|\bma_m\|^2+\sigma_m^2) + 1
\end{align*}
and the matrix inversion lemma to simplify 
\begin{align*}
\widetilde\bK & = (\gamma\tilde\bma_m\tilde\bma^H_m+\bC^{-1}_{\tilde\bmx})^{-1} \\
& = \bC_{\tilde\bmx} - \frac{\gamma\bC_{\tilde\bmx}\tilde\bma_m\tilde\bma^H_m\bC_{\tilde\bmx}}{\gamma\tilde\bma^H_m\bC_{\tilde\bmx}\tilde\bma_m+1} \\
& = \bC_{\tilde\bmx} - \frac{\gamma\bC_{\tilde\bmx}\tilde\bma_m\tilde\bma^H_m\bC_{\tilde\bmx}}{\gamma(\sigma_x^2\|\bma_m\|^2+\sigma_m^2)+1}.
\end{align*}
By using these two simplifications, we have
\begin{align*}
&   \Ex{}{\exp(-\gamma|\tilde\bma_m^H\tilde\bmx|^2)\tilde{x}_n \tilde{x}_{n'}^*}   \\
 & = \frac{1}{ \gamma(\sigma_x^2\|\bma_m\|^2+\sigma_m^2) + 1} \\
 & \qquad \times \left[ \bC_{\tilde\bmx} - \frac{\gamma\bC_{\tilde\bmx}\tilde\bma_m\tilde\bma^H_m\bC_{\tilde\bmx}}{\gamma(\sigma_x^2\|\bma_m\|^2+\sigma_m^2)+1} \right]_{n,n'}
\end{align*}
and since we are only interested in the upper $N\times N$ part of the matrix~$\widetilde\bK$, we have
\begin{align*}
&   \Ex{}{\exp(-\gamma|\bma_m^H\bmx+[\bme^z]_m|^2) x_n x_{n'}^*}  \\
 & = \frac{1}{ \gamma(\sigma_x^2\|\bma_m\|^2+\sigma_m^2) + 1} \\
 & \qquad \times \left[ \sigma^2_x\bI_{N} - \frac{\gamma\sigma_x^4\bma_m\bma^H_m}{\gamma(\sigma_x^2\|\bma_m\|^2+\sigma_m^2)+1} \right]_{n,n'} \\
  & = \frac{1}{ \gamma[\bC_\bmz]_{m,m} + 1}\left[ \sigma^2_x\bI_{N} - \frac{\gamma\sigma_x^4\bma_m\bma^H_m}{\gamma[\bC_\bmz]_{m,m}+1} \right]_{n,n'} 
\end{align*}
since for our assumptions 
\begin{align*}
\sigma_x^2\|\bma_m\|^2+\sigma_m^2 = [\bC_\bmz]_{m,m}.
\end{align*}
In compact matrix form, we have
\begin{align*}
&  \Ex{}{\exp(-\gamma|\bma_m^H\bmx+[\bme^z]_m|^2)  \bmx \bmx^H }   \\
 & \qquad = \frac{1}{ \gamma[\bC_\bmz]_{m,m} + 1}\left( \sigma^2_x\bI_{N} - \frac{\gamma\sigma_x^4\bma_m\bma^H_m}{\gamma[\bC_\bmz]_{m,m}+1} \right).
\end{align*}
As a second step, we use definition \fref{eq:pquantity} and obtain
\begin{align*}
 \Ex{}{\exp(-\gamma [\bme^y]_{m})} = [\bmp_\gamma]_m.
\end{align*}
By combining both steps, we obtain
\begin{align*}
\bV_m & =  \frac{[\bmp_\gamma]_m}{ \gamma[\bC_\bmz]_{m,m} + 1}\left( \sigma^2_x\bI_{N} - \frac{\gamma\sigma_x^4\bma_m\bma^H_m}{\gamma[\bC_\bmz]_{m,m}+1} \right)  \\
& \qquad - \frac{[\bmp_\gamma]_m}{ \gamma[\bC_\bmz]_{m,m}+1 } \sigma_x^2 \bI_N \\
&  = - \frac{\gamma\sigma_x^4[\bmp_\gamma]_m}{(\gamma[\bC_\bmz]_{m,m}+1)^2} \bma_m\bma^H_m,
\end{align*}
which is what we desperately wanted to show.
